\documentclass[journal,twocolumn]{IEEEtran}
\usepackage{cite}

\usepackage[hidelinks]{hyperref}

\ifCLASSINFOpdf
   \usepackage[pdftex]{graphicx}
   \graphicspath{{./}}
   \DeclareGraphicsExtensions{.pdf,.jpeg,.png,.eps}
\else
\fi
\usepackage{amsmath}
\usepackage{amsfonts}
\usepackage{amsthm}
\usepackage{amssymb}

\usepackage{braket}
\usepackage{xcolor-solarized}

\usepackage[nameinlink,capitalise]{cleveref}

\newtheorem{theorem}{Theorem}
\newtheorem{lemma}[theorem]{Lemma}
\newtheorem{corollary}[theorem]{Corollary}

\Crefname{lemma}{Lemma}{Lemma}
\Crefname{theorem}{Theorem}{Theorem}
\Crefname{corollary}{Corollary}{Corollary}
\newtheorem*{theorem*}{Main Theorem}
\newtheorem*{conjecture*}{Conjecture}
\newtheorem*{corollary*}{Corollary}

\newcommand{\del}{\partial}
\newcommand{\im}{\operatorname{im}}
\newcommand{\id}{\operatorname{id}}
\newcommand{\Tot}{\operatorname{Tot}}

\usepackage{array}

\usepackage{tikz}
\usepackage{tikz-cd}

\ifCLASSOPTIONcompsoc
  \usepackage[caption=false,font=normalsize,labelfont=sf,textfont=sf]{subfig}
\else
  \usepackage[caption=false,font=footnotesize]{subfig}
\fi
\usepackage{url}

\usepackage{xcolor}

\newcommand{\dist}{\operatorname{dist}}

\begin{document}
%
\title{Balanced Product Quantum Codes}
%
%
%

\author{Nikolas~P.~Breuckmann
	and~Jens~N.~Eberhardt
	
\thanks{Nikolas P. Breuckmann, University College London, United Kingdom, \href{mailto:n.breuckmann@ucl.ac.uk}{n.breuckmann@ucl.ac.uk}}%
\thanks{Jens N. Eberhardt, University of Bonn, Germany, \href{mailto:mail@jenseberhardt.com}{mail@jenseberhardt.com}}%
}

\maketitle

\begin{abstract}
This work provides the first explicit and non-random family of $[[N,K,D]]$ LDPC quantum codes which encode $K \in \Theta(N^\frac{4}{5})$ logical qubits with distance $D \in \Omega(N^\frac{3}{5})$.
The family is constructed by amalgamating classical codes and Ramanujan graphs via an operation called {\em balanced product}.

Recently, Hastings--Haah--O'Donnell and Panteleev--Kalachev were the first to show that there exist families of LDPC quantum codes which break the $\operatorname{polylog}(N)\sqrt{N}$ distance barrier.
However, their constructions are based on probabilistic arguments which only guarantee the code parameters with high probability whereas our bounds  hold unconditionally.

{Further, balanced products allow for non-abelian twisting of the check matrices, leading to a construction of LDPC quantum codes that can be shown to have $K\in \Theta(N)$ and that we conjecture to have linear distance $D\in \Theta(N)$.}
\end{abstract}


%
\IEEEpeerreviewmaketitle

\section{Introduction}
%
%
%
%
%
%

\IEEEPARstart{T}{he}  construction of low-density parity check (LDPC) quantum codes has some unique challenges in comparison to classical LDPC codes~\cite{breuckmann2021ldpc}.
While there are classical codes with constant encoding rate and linear distance, so-called {\em good codes}, no equivalent statement is known for LDPC quantum codes.
Even more severely, no LDPC quantum codes with distance significantly larger than~$\sqrt{N}$ were known to exist until recently.\footnote{There is a nice construction due to Bacon et.~al.~\cite{bacon2017sparse} which gives subsystem codes with distance $D\in \Omega(n^{1-\epsilon})$. It has sparse gauge operators but it is not LDPC. There are other interesting constructions using a mapping of quantum circuits to Hamiltonians~\cite{breuckmann2014space} where the LDPC condition is slightly weakened, see~\cite{bohdanowicz2019good}.}

For a long time, the best known distance was achieved by Freedman--Meyer--Luo in~\cite{freedman2002z2}.
Their codes are derived from fiber bundles over the torus~$T^2$.
The fibers are hyperbolic surfaces twisted along geodesics, giving rise to a quantum code family with distance $\sqrt[4]{\log N} \sqrt{N}$.
Evra--Kaufman--Z\'emor~\cite{ramanujancomplexesone} and Kaufman--Tessler~\cite{ramanujancomplexestwo} further improved this to $\sqrt{\log N} \sqrt{N}$ and $\log^k(N) \sqrt{N}$, respectively.

Recently, the apparent $\operatorname{polylog}(N)\sqrt{N}$ distance barrier was broken.
Hastings--Haah--O'Donnell~\cite{hastings2020fiber} introduced {\em fiber bundle codes} which have a random classical code as base and  a cycle graph as fiber.
They obtain $K\in \Theta(N^\frac{3}{5})$ and $D\in \Omega(N^\frac{3}{5}/\operatorname{polylog}(N))$.
Panteleev--Kalachev~\cite{panteleev2019degenerate,panteleev_almostlinear} obtained codes with almost linear distance bounds using a similar construction called {\em lifted product}, where $K\in \Theta(N^\alpha \log(N))$ and $D\in \Omega(N^{1-\alpha/2}/\log(N))$ for $0\leq \alpha < 1$.
Both constructions are probabilistic and not explicit.
Here, \emph{explicit} is a technical term which means that the codes can be generated in polynomial time by a deterministic algorithm, so that the bounds on the code parameters are guaranteed to hold.

In this manuscript we explain how to construct families of quantum codes with exceptional asymptotic properties in an explicit and highly symmetrical way.
To this end, we mimic the well-known {\em balanced product} $X\times_{H} Y$ of two topological spaces~$X$ and~$Y$ which admit the action of a group~$H$~\cite{james2012general}.
Balanced products may be familiar to physicists as they appear in the construction of vector bundles associated to principal bundles.

In analogy to the topological case, we introduce \emph{balanced product codes}~$C\otimes_H D$ which are constructed from two (quantum) codes~$C$ and~$D$ that have a common symmetry realized by the action of a group~$H.$

In the trivial case $H=\{1\},$ the balanced product $X\times_H Y$ coincides with the cartesian product $X\times Y$ and, similarly, balanced product codes are just hypergraph product codes~\cite{tillich2013quantum}.
For a non-trivial~$H$ the number of physical qubits is reduced by a factor of up to~$|H|$ while, in certain cases, the distance is preserved. Hence, balanced product codes can achieve better relative distances than hypergraph product codes.

From a topological point of view, the balanced product~$X\times_H Y$ is a fiber bundle over the base~$X/H$ with fiber~$Y$ and can be understood as \emph{twisting} $Y$ along $X/H$ in the cartesian product $X/H\times Y.$  Similarly, balanced product codes can be obtained from hypergraph product codes by introducing twists in their checks. These twists can result in an increased relative distance.

When $H$ is a cyclic group, balanced product codes overlap with the fiber bundle codes considered in~\cite{hastings2020fiber} and lifted product codes from~\cite{panteleev2019degenerate,panteleev_almostlinear}. 
In this case, one can prove lower bounds on the distance which is we why focus on this class of balanced product codes.
We note, however, that balanced product codes are more general than these constructions and also work for non-abelian groups~$H$ and non-free group actions. In fact, we conjecture that good LDPC quantum codes can be constructed using non-abelian balanced products.

Concretely, we will focus on balanced products of good classical codes and repetition codes which have an action of a cyclic group~$H$.
To obtain good classical codes with large cyclic symmetry, we employ a well-known construction due to Sipser--Spielman~\cite{sipser1996expander} which combines certain families of graphs, called expander graphs, with small codes of fixed size.

We consider two families of expander graphs:
The first are Ramanujan graphs arising as Cayley graphs of the projective linear group~$\operatorname{PGL}(2,q)$ and were introduced by Lubotzky--Phillips--Sarnak~\cite{lubotzky1988ramanujan}.
The second, similar, construction uses tessellations and finite coverings of hyperbolic surfaces~\cite{breuckmann2020single}.
Both constructions admit a cyclic symmetry of known order.
While we ultimately employ the Lubotzky--Phillips--Sarnak (LPS) construction, we note that hyperbolic surfaces are extremely flexible and could be more suitable for practical applications.

By carefully choosing the code parameters of the factors of the balanced product we obtain our main result.
\begin{theorem*}[see \Cref{thm:explicitfamily}]
There is an explicit family of $[[N,K,D_X,D_Z]]$ LDPC quantum codes encoding $K \in \Theta(N^\frac{2}{3})$ logical qubits with $X$-distance $D_X \in \Omega(N^\frac{1}{3})$ and $Z$-distance $D_Z \in \Theta(N)$.
\end{theorem*}
Applying the distance balancing procedure of~\cite{hastings2016weight} and~\cite{ramanujancomplexesone} we obtain the following corollary.
\begin{corollary*}[see \Cref{cor:explicitfamily}]
There is an explicit family of $[[N,K,D]]$ LDPC quantum codes encoding $K \in \Theta(N^\frac{4}{5})$ logical qubits with distance $D \in \Omega(N^\frac{3}{5})$.
\end{corollary*}
We note that due to a technical condition some members of the constructed code family may turn out to be a subsystem codes.
However, we stress that they are LDPC with respect to the stabilizer checks.
Therefore, they do not share the same draw-backs as previously known subsystem codes such as~\cite{bacon2017sparse}. 

By taking expander codes based on Cayley graphs of different groups it is likely that we can obtain parameters $K\in \Theta(N^\alpha )$ and $D\in \Omega(N^{1-\alpha/2})$ for different values $\alpha < 4/5.$
Unfortunately, we are not aware of suitable examples.

We believe that our methods can be used to construct families of good LDPC quantum codes, see \Cref{sec:conclusion}.

{\em Remark}---This manuscript arose with the intention of giving an explicit construction and a more conceptual understanding of the fiber bundle codes of Hastings--Haah--O'Donnell \cite{hastings2020fiber}.
By adapting their distance arguments we were able to explicitly construct balanced product codes with $K\in \Theta(N^{\frac{3}{5}})$ and $D\in\Omega(N^{\frac{3}{5}})$ improving their distance bounds by a factor of $\operatorname{polylog}(N)$.
However, shortly before we intended to publish, Panteleev--Kalachev showed that even better distance bounds are attainable for lifted product codes \cite{panteleev_almostlinear}.
We make use of a result of theirs to further improve our bounds.

\subsection*{Summary}
In \Cref{sec:complexes} we give a homological description of (quantum) codes and their products.
We introduce useful notions, such as double complexes and total complexes, which allow us to describe products of quantum codes from an abstract level.
These tools allow us to compute the \emph{encoding rate} of the balanced product codes constructed in this manuscript. 
Moreover, we use them to give new and streamlined proofs of results of Hastings--Haah--O'Donnell on fiber bundle codes~\cite{hastings2020fiber}.

In \Cref{sec:expander} we construct highly symmetrical good classical LDPC codes which will serve as factors of our balanced product codes. 
To this end, we use an idea due to Sipser--Spielman which builds good classical LDPC codes, called expander codes, from expander graphs. We study the combinatorial properties of expander codes in detail, as they will be used to bound the \emph{relative distance} of the resulting balanced product codes. To obtain expander codes with large symmetry groups, we consider LPS expander graphs which are Cayley graphs of the groups $\operatorname{PGL}(2,q)$ and, alternatively, graphs derived from hyperbolic geometry.

In \Cref{sec:balanced_products} we introduce the general framework for balanced product codes. 
We give a topological motivation and explain the relation of balanced products to fiber bundles. 
We explain how this topological notion can be discretized in the setting of graphs. 
By linearization, we obtain a notion of balanced product of vector spaces and chain complexes, thus yielding a general definition of balanced product codes. 
We introduce a Künneth formula and compare balanced products to fiber bundle and lifted product codes. 
Next, we discuss balanced products of expander codes and (cyclic) graphs in detail. 
We then discuss a subsystem version of balanced product codes. At the end of the section, we adapt Panteleev--Kalachev's distance theorems for lifted product codes~\cite{panteleev_almostlinear} to our setting. 

Finally, in \Cref{sec:explicit} we prove our main theorem by applying the balanced product construction from \Cref{sec:balanced_products} to the expander codes of \Cref{sec:expander} and cyclic repetition codes.
We show that the expander codes built from LPS expander graphs admit large cyclic symmetry which allow us to form the balanced product with a cyclic repetition code. 
We show that, after carefully choosing parameters, the resulting family of codes fulfills the bounds of our main theorem.

In the conclusion, \Cref{sec:conclusion}, we formulate our conjecture regarding the construction of a family of quantum LDPC codes that we believe to be \emph{good}, i.e. having parameters scaling as $K\in \Theta(N)$ and $D\in \Theta(N)$.

\section{(Quantum) Codes and Chain Complexes}\label{sec:complexes}
In this section we give a homological perspective on classical and quantum codes which puts them on equal footing and makes them amenable to tools from homological algebra and topology.
In particular, we explain how to construct quantum codes from cell complexes, double complexes and the fiber bundle construction of Hastings--Haah--O'Donnell \cite{hastings2020fiber}.
For the latter we provide some new perspectives which allows for more streamlined proofs of some of their homological results.

Note that for the remainder of this manuscript we will only consider vector spaces over the field $\mathbb{F}_2$.

\subsection{Complexes of Vector Spaces}
A \emph{chain complex} of vector spaces $C=(C_{\bullet},\del)$ is a set of vector spaces $C_{i}$ equipped with linear maps, called \emph{differentials},
\begin{align*}
\del_{i}&:C_{i}\to C_{i-1}
\end{align*}
that square to zero
\begin{align*}
\del_{i-1} \del_i=0 
\end{align*}
In the following we will mostly drop the indices from the notation, so for example write $\del^2=0.$

We denote by 
\begin{align*}
Z_i(C) &= \ker \del_i\subset C_i\\
B_i(C) &= \im \del_{i+1}\subset C_i\\
H_i(C) &=Z_i(C)/B_i(C)
\end{align*}
the  \emph{$i$-cycles}, \emph{$i$-boundaries} and the $i$-th \emph{homology} of the complex $C$, respectively. Elements in $C_i$ are called \emph{$i$-chains}.

For most complexes we will considered, the spaces $C_i$ are equipped with a canonical bases of so called $i$-cells. Hence, there is a natural scalar product  
$$\langle v, w\rangle \in \mathbb{F}_2\text{ for } v,w\in C_i.$$
In this case, we will take the freedom to write $C^i$ for the linear dual of $C_i$ and identify $$C_i=C^i$$ using the scalar product. We refer to elements in $C^i$ as cochains.
We denote the transposed maps of the differential $\del$ by
\begin{align*}
\delta^{i}=(\del_{i-1})^{tr}&:C^{i}\to C^{i+1}
\end{align*}
and define by
\begin{align*}
Z^i(C) &= \ker \delta^i\subset C^i\\
B^i(C) &= \im \delta^{i-1}\subset C^i\\
H^i(C) &=Z^i(C)/B^i(C)
\end{align*}
the  \emph{$i$-cocycles}, \emph{$i$-coboundaries} and the $i$-th \emph{cohomology} of the complex $C.$
The scalar product on $C_i$ and $C^i$ induces a well-defined and non-degenerate pairing of $H_i(C)$ and $H^i(C)$ since
$B^i(C)=Z_i(C)^\perp.$
The scalar product hence induces an isomorphism 
$$H_i(C)\cong H^i(C).$$

\subsection{(Quantum) Codes from Complexes}\label{sec:classicalandquantumcodes}

\subsubsection{Classical codes}
A classical, linear, binary code~$\mathcal{C}$ is a subspace of~$\mathbb{F}_2^n$.
The number of encodable bits in $\mathcal{C}$ is its dimension $k = \dim \mathcal{C}$.
The smallest Hamming weight of any non-zero element of~$\mathcal{C}$ is called the {\em distance} of~$\mathcal{C}$ and denoted by~$d$.
We call a classical code of size~$n$, dimension~$k$ and distance~$d$ an $[n,k,d]$-code.

A linear code can be specified as the kernel of a {\em parity check matrix}.
This allows us to write a classical code $\mathcal{C}$ in terms of a chain complex $C$ as follows.
Let 
\begin{align*}
C &= (C_{1}\stackrel{\del^{C}}{\to}C_{0})
\end{align*}
such that $C_1=\mathbb{F}_2^n$ and $\del^{C}$ is the parity check matrix of the code~$\mathcal{C}$.
Then~$\mathcal{C}$ is the space of $1$-cycles $\mathcal{C} =H_1(C) = Z_1(C) = \ker \del^{C}_1$.
Furthermore, the space of checks acting on the code is the space of $0$-chains.

Any code~$\mathcal{C}$ gives rise to a chain complex $C$ in this way.
On the other hand, any chain complex $C=(C_{\bullet},\del)$ in which every vector space $C_i$ is equipped with a basis contains classical codes; we simply pick an index~$i$ and take $\mathcal{C} = Z_i^C = \ker \del^{C}_i$.

From now on we will abuse notation and use the same symbol~$C$ for both the code as a subspace of~$\mathbb{F}_2^n$ and the chain complex  $(C_{\bullet},\del)$.
It will always be clear from context which interpretation to choose.

Anyone with a sole interest in classical codes will not benefit from viewing them through the lens of homology.
It is useful to us because it puts classical codes on equal footing with quantum codes as we will see below.

\subsubsection{Stabilizer quantum codes}\label{sec:stabilizerquantumcodes}
Stabilizer quantum codes are defined by their {stabilizer group}~$\mathcal{S}$ which is an abelian subgroup of the Pauli group operating on $n$ qubits that does not contain~$-I$.
The code space is the subspace of the Hilbert space of $n$ qubits which is point-wise stabilized by~$\mathcal{S}$.
Pauli operators which stabilize the code space (as a whole) but are not in~$\mathcal{S}$ are called {\em logical operators}.

A special class of stabilizer codes are those where the stabilizer group~$\mathcal{S}$ has a set of generators which operate as all-$X$ or all-$Z$.
They are called {\em CSS quantum codes}.
We note in passing that an arbitrary $[[n,k,d]]$ stabilizer code can be mapped onto an $[[4n,2k,2d]]$ CSS code~\cite{bravyi2010majorana}.
If we are only interested in asymptotic scaling of the code parameters it is thus sufficient to only consider CSS codes.

CSS codes are in bijection with chain complexes of length two.
This is because the $X$/$Z$-type generators can be mapped onto matrices $H_X$/$H_Z$.
The condition that stabilizers need to commute is equivalent to demanding $H_Z^T H_X = 0$.
Hence, a CSS code is nothing but a chain complex
\begin{center}
	\begin{tikzcd}
		C_{1} \arrow[r,"H_Z^T"] & C_0 \arrow[r,"H_X"] & C_{-1}.
	\end{tikzcd}
\end{center}

As we have seen previously for classical codes the reverse is also true:
We can always obtain a quantum code from a chain complex $C=(C_{\bullet},\del)$  in which every vector space $C_i$ is equipped with a basis of $i$-cells.
Namely, we pick some index~$i$ and take the matrices~$H_Z^T$ and~$H_X$ representing the boundary operators in
\begin{center}
	\begin{tikzcd}
		C_{i+1} \arrow[r,"\del_{i+1}"] & C_i \arrow[r,"\del_{i}"] & C_{i-1}
	\end{tikzcd}
\end{center}
to define a CSS code.

The parameters of the resulting CSS code are given by
\begin{align*}
n = \dim C_i \\
k = \dim H_i(C) 
\end{align*}
and the distance~$d$ is the minimum weight of a representative of a non-trivial element of~$H_i(C)$ or~$H^i(C)$.
See \cite{breuckmann2017homological} for more details.

We call a family of stabilizer codes {\em low-density parity-check codes (LDPC)} if there exists some $w\in \mathbb{N}$ such that for each member of the family there exists a set of generators of the stabilizer group such that each generator has weight at most~$w$ and each qubit is only in the support of at most~$w$ checks.

\subsubsection{Subsystem codes}
Consider a CSS code where the qubits are identified with the $i$-cells of a chain complex~$C$, as explained above.
A {\em subsystem CSS code} is a CSS code in which we only utilize a subset of the logical qubits for information storage.
The other logical degrees of freedom are being downgraded to so-called {\em gauge qubits}.

From the perspective of homology this corresponds to splitting the $i$-th homology of~$C$ into a direct sum $H_i = H_i^\mathcal{L} \oplus H_i^\mathcal{G}$ where $H_i^\mathcal{L}-\{0\}$ is in bijection to the logical $Z$-operators (up to stabilizers) and $H_i^\mathcal{G}-\{0\}$ is in bijection to the non-trivial gauge $Z$-operators (up to stabilizers).
The pairing of the homology group~$H_i$ with the cohomology group~$H^i$ induces the compatible splitting $H^i= H^i_\mathcal{L} \oplus H^i_\mathcal{G}$.

We call the minimum weight of any representative of a class in $H_i^\mathcal{L}$ ($H^i_\mathcal{L}$) the {\em bare $Z$-distance} ({\em bare $X$-distance}).
Note that it is generally possible to reduce the weight of the representatives of classes in~$H_i^\mathcal{L}$ and~$H^i_\mathcal{L}$ by adding representatives of~$H_i^\mathcal{G}$ and~$H^i_\mathcal{G}$.
This makes it necessary to define another notion of distance:
The {\em  $Z$-distance} $d_Z$ of a subsystem CSS code is the largest integer such that for all $[z]\in H_i$ with the property that the restriction to the logical part is non-zero $[z]|_{H_i^\mathcal{L}} \neq 0$ we have that $|z| \geq d_Z$.
The {\em  $X$-distance}~$d_X$ is defined equivalently.
The {\em  distance}~$d$ of a subsystem code is $d = \min \{d_X,d_Z\}$.\footnote{What we call the distance of a subsystem code is also sometimes called {\em dressed distance} in the literature.}

A family of subsystem codes is called an {\em LDPC subsystem code} if there exists some $w\in \mathbb{N}$ such that for each member of the family there exists a set of generators of the {\em stabilizer group} such that each generator has weight at most~$w$ and each qubit is only in the support of at most~$w$ checks.
We note that families of codes exists for which the stabilizer generators can be measured indirectly by measuring a generating set of the gauge operators and there are cases where the latter all have bounded support and each qubit is acted upon by a bounded number of gauge operators.
There exists a nice construction due to Bacon et.~al.~\cite{bacon2017sparse} which gives subsystem codes with dressed distance $d \in \Omega(n^{1-\epsilon})$ and which has sparse gauge operators but which is not LDPC.
For more background on quantum codes see~\cite{terhal2015quantum} and for the application of homology to (quantum) codes see \cite{bombin2007homological,leslie2014hypermap,breuckmann2017homological,forney2018codes}.

\subsection{Complexes from Spaces}
The origin of chain complexes and homology lie in algebraic topology.
We briefly recall how to associate a chain complex to a cell complex.

Cell complexes are discrete and combinatorial analogues of topological spaces and arise as cellulations thereof.
A regular cell complex~$X$ is a finite union of cells~$\sigma$ which are glued in a nice way. 
An $n$-dimensional cell in~$X$ is called an \emph{$n$-cell}. 
We denote by $X_{n}$ the set of all $n$-cells in~$X$ and by~$\del\sigma$ the set of all $(n-1)$-cells in the boundary of the $n$-cell~$\sigma$.
 
For example, consider the cell complex $X=C_\ell$ with $1$-cells $\sigma_{1},\dots,\sigma_\ell$ and $0$-cells $\tau_{1},\dots,\tau_{\ell}$ such that the boundary of $\sigma_{i}$ is $\del\sigma_{i}=\{\tau_{i},\tau_{i+1}\},$ where we take indices modulo $\ell$. 
Then~$C_\ell$ corresponds to a celluation of a circle~$S^1$ into~$\ell$ pieces. 
Equivalently,~$C_\ell$ can be viewed as a cycle graph, where $1$-cells and $0$-cells correspond to edges and vertices, respectively.

To a cell complex $X$ one can associate a chain complex $C(X)=(C_\bullet(X),\del)$ via
\begin{align*}
C_{i}(X)=\bigoplus_{\sigma \in X_{i}}\mathbb{F}_2 \sigma\text{ and }
\del(\sigma)=\sum_{\tau\in \del\sigma}\tau.
\end{align*}
With this, we can define the \emph{homology} of a cell complex $X$ via
$$H_{i}(X)=H_{i}(C(X)).$$

In our previous example it is easy to see that $\dim H_1(C_\ell) = 1$.
Furthermore, interpreting the complex $C(C_\ell)$ as a classical code as described in \Cref{sec:classicalandquantumcodes} yields the repetition code.
Another example is the $\ell\times \ell$-torus~$T_\ell$ which is a cell complex of dimension~$2$.
Here $\dim H_1(T_\ell) = 2$ and the quantum code associated to the chain complex $C(T
_\ell)$ yields the famous toric code with parameters~$[[2\ell^2,2,\ell]]$.

\subsection{Double Complexes and Total Complexes}\label{sec:doublecomplexes}
In a wide variety of situations, chain complexes arise as the \emph{total complex} of a \emph{double complex}.
A double complex $E=(E_{\bullet,\bullet},\del^{v},\del^{h})$ is an array of vector spaces $E_{p,q}$ equipped with vertical and horizontal maps
\begin{align*}
\del_{p,q}^{v}&:E_{p,q}\to E_{p,q-1}\text{ and }\\
\del_{p,q}^{h}&:E_{p,q}\to E_{p-1,q}
\end{align*}
such that $\del^{v}$ and $\del^{h}$ are commuting differentials
\begin{align*}
(\del^{v})^{2}=(\del^{h})^{2}=0 \text{ and } \del^{v}\del^{h}=\del^{h}\del^{v},
\end{align*}
see \Cref{fig:doublecomplexcommute}.
Here, and in the following, we will mostly drop the indices of differentials from the notation.
\begin{figure}
\begin{center}
\begin{tikzcd}
 {E_{p,q}} \arrow[d, "{\del_{p,q}^v}"'] \arrow[r, "{\del_{p,q}^h}"] & {E_{p-1,q}} \arrow[d, "{\del_{p-1,q}^v}"] \\
 {E_{p,q-1}} \arrow[r, "{\del_{p,q-1}^h}"']  & {E_{p-1,q-1}}   &       
\end{tikzcd}
\end{center}
\caption{A square in a double complex $E$ which is required to commute.}
\label{fig:doublecomplexcommute}
\end{figure}

It is often useful to interpret a double complex as a ``complex of complexes''. 
Namely, each column of $C$ forms a vertical complex using the vertical differential $\del^{v}$ whereas the horizontal differentials $\del^{h}$ define a horizontal complex of these vertical complexes.\subsubsection{Total Complex}
A double complex $E$ can be collapsed into a complex by ``summing over the diagonals''. 
The resulting \emph{total complex} $\Tot(E)$ is the complex defined via
\begin{gather*}
\Tot(E)_{n}=\bigoplus_{p+q=n} E_{p,q}
\end{gather*}
and the differential is given by $\del=\del^{v}+\del^{h}.$ In fact
\begin{align*}
\del^{2}&=(\del^{v}+\del^{h})^{2}
\\&=(\del^{v})^{2}+(\del^{v}\del^{h}+\del^{h}\del^{v})+(\del^{h})^{2}=0
\end{align*}
using that $E$ is a double complex.
\subsubsection{Tensor Product and Hypergraph Product}
A fundamental example of a double complex is obtained by forming a tensor product of two complexes. Let $C=(C_{\bullet},\del^{C})$ and $D=(D_{\bullet},\del^{D})$ be two complexes. 
Then the \emph{tensor product  double complex}  $C\boxtimes D$ of $C$ and $D$ is defined via
\begin{gather*}
(C\boxtimes D)_{p,q}=C_{p}\otimes D_{q},\\
\del^{v}=\del^{C}\otimes\id_{D} \text{ and } \del^{h}=\id_{C}\otimes\del^{D}.
\end{gather*}
\begin{figure}
\begin{center}
\begin{tikzcd}
C_{p}\otimes D_{q} \arrow[d, " \id\otimes\del^{D}_{q}"'] \arrow[r, "\del^{C}_{p-1}\otimes\id"] &C_{p-1}\otimes D_{q}\arrow[d, "\id\otimes\del^{D}_{q-1}"]  \\
 C_{p}\otimes D_{q-1} \arrow[r, "\del^{C}_{p-1}\otimes\id"']          & C_{p-1}\otimes D_{q-1}      
\end{tikzcd}
\end{center}
\caption{A square in the tensor product double complex $C\boxtimes D.$}
\end{figure}
The \emph{tensor product complex} of two complexes $C$ and $D$
\begin{gather*}
C\otimes D = \Tot( C\boxtimes D)
\end{gather*}
is the total complex of their tensor product double complex. 

The special case where where $C$ and $D$ are 1-complexes are also known as {\em hypergraph products} which were introduced in~\cite{tillich2013quantum}.
This was generalized for $C$ being a complex of arbitrary finite length in~\cite{zeng2019higher}.

\subsubsection{Homology of Total Complexes}
In general, the computation of the homology of a total complex $\Tot(E)$ can be quite subtle and subject of the powerful mathematical formalism of \emph{spectral sequences}, which we briefly review in Appendix \ref{sec:spectralsequence}. 
However, there are certain situations, in which the homology of a double complex can be computed from the homology of its vertical and horizontal complexes.

The homology of the complex  $C\otimes D$ is subject of the \emph{Künneth formula}.
\begin{theorem}[Künneth formula]
There is a natural isomorphism
\begin{gather*}
H_{n}(C\otimes D) \cong\bigoplus_{p+q=n} H_{p}(C)\otimes H_{q}(D).
\end{gather*}
\end{theorem}

Another simple example, which will be of interest for us, is the situations in which $E$ is a $2\times 2$ complex, that is, $E_{p,q}=0$ for all $p,q\neq\{0,1\}.$ In this case, the homology of $\Tot(E)$ can be computed by first taking homology in the vertical and then horizontal direction. 

First, one takes homology along the differentials $\del^{v}$, see \Cref{fig:verticaldifferentialindoublecomplex}. 
The horizontal differential $\del^{h}$ induces a differential on resulting homology groups $H_{q}(E_{p,\bullet},\del^{v}),$ see \Cref{fig:inducedhorizontaldifferential}.
Secondly, one takes homology along these induced horizontal differentials to obtain homology groups
$$H_{p}(H_{q}(E_{\bullet,\bullet},\del^{v}),\del^{h}).$$ The homology of the total complex is comprised of these groups and one obtains an analogue of the Künneth formula.
\begin{theorem}\label{thm:smalldoublecomplex}
If $E$ is a $2\times 2$-complex there is a natural isomorphism
\begin{gather*}
H_{n}(\Tot(E)) \cong\bigoplus_{p+q=n} H_{p}(H_{q}(E_{\bullet,\bullet},\del^{v}),\del^{h}).
\end{gather*}
\end{theorem}
\begin{proof}
All differentials  on the second page of the spectral sequence of the double complex $E$ vanish since either their domain or codomain is zero. See Appendix \ref{sec:spectralsequence}.
\end{proof}
\begin{figure}
\begin{center}
\begin{tikzcd}
 {E_{1,1}} \arrow[d, "{\del^v}"'] & {E_{0,1}} \arrow[d, "{\del^v}"] \\
 {E_{1,0}} & {E_{0,0}}   &       
\end{tikzcd}
\end{center}
\caption{The vertical differential in the double complex $E_{p,q}$ }
\label{fig:verticaldifferentialindoublecomplex}
\end{figure}
\begin{figure}
\begin{center}
\begin{tikzcd}
H_{1}(E_{1,\bullet},\del^{v}) \arrow[r, "{\del^h}"] & H_{1}(E_{0,\bullet},\del^{v})\\
H_{0}(E_{1,\bullet},\del^{v})\arrow[r, "{\del^h}"']  &H_{0}(E_{0,\bullet},\del^{v})  &       
\end{tikzcd}
\end{center}
\caption{The induced horizontal differential on the vertical homology groups $H_{q}(E_{p,\bullet},\del^{v}).$}
\label{fig:inducedhorizontaldifferential}
\end{figure}

\subsection{Fiber Bundle Codes}\label{sec:fiberbundlecodes}
Fiber bundle codes where introduced by Hastings et al. in~\cite{hastings2020fiber} to demonstrate a construction of quantum codes breaking the $\sqrt{N}\operatorname{polylog}(N)$ distance bounds for the first time.
The idea behind fiber bundle codes is to introduce a \emph{twist} in the (horizontal) differentials in the tensor product double complex, in order to increase the distance of the resulting code. 

For simplicity, let us consider two 1-complexes
\begin{align*}
B&=(B_{1}\stackrel{\del^{B}}{\to}B_{0})\text{ and}\\
F&=(F_{1}\stackrel{\del^{F}}{\to}F_{0}),
\end{align*}
we refer to as  \emph{base} and \emph{fiber} respectively. Denote by $\operatorname{Aut}(F)$ the finite group of linear automorphisms of the complex $F,$ that is, linear automorphisms of $F_{1}$ and $F_{0}$ that commute with the differential $\del^{F}.$ Further, denote basis vectors of $B_{i}$ by $b^{i}$ and write $b^{0}\in \del^{B}b^{1}$ if $b^{0}$ appears with a nonzero coefficient in~$\del^{B}b^{1}$.

The idea is to twist the horizontal differentials in the tensor product double complex $B\boxtimes F$ by a \emph{connection} or \emph{twist} $\varphi$ which is a choice of  an automorphism
$\varphi(b^{1},b^{0})\in \operatorname{Aut}(F)$
to every pair $(b^{1},b^{0})$ such that $b^{0}\in\del^{B}b^{1}.$ Intuitively, the connection $\varphi$ describes how the fiber~$F$ varies over the base~$B$, hence its name.

The \emph{fiber bundle double complex} $B\boxtimes_{\varphi}F$ is  
\begin{center}
\begin{tikzcd}
B_{1}\otimes F_{1} \arrow[d, " \id\otimes\del^{F}"'] \arrow[r, "\del_{\varphi}"] &B_{0}\otimes F_{1}\arrow[d, "\id\otimes\del^{F}"]  \\
 B_{1}\otimes F_{0} \arrow[r, "\del_{\varphi}"']          & B_{0}\otimes F_{0}
\end{tikzcd}
\end{center}
where
\begin{align*}
\del_{\varphi}(b^{1}\otimes f)=\sum_{b^{0}\in\del^{B}b^{1}}b^{0}\otimes \varphi(b^{1},b^{0})(f).
\end{align*}
Chains in $B_{1}\otimes F_{0},$ resp. $B_{0}\otimes F_{1},$ are referred to as \emph{horizontal}, resp. \emph{vertical}, where we think of the base and fiber stretching out in the horizontal, resp. vertical, direction.

The \emph{fiber bundle code} is given by the total complex $$B\otimes_{\varphi} F=\Tot(B\boxtimes_{\varphi}F).$$
It is easy to see that $B\otimes_{\varphi} F$ is really a double complex; $\del_{\varphi}$ and  $\id\otimes\del^{F}$ commute since  $\varphi(b^{1},b^{0})$ commutes with $\del^{F}.$
Note that when $\varphi=1$ then $\del_{\varphi}=\del^{B}\otimes \id$ and $B\otimes_{\varphi} F=B\otimes F$ is the usual tensor product complex. 

Under mild conditions, the homology of the fiber bundle complex $B\otimes_{\varphi} F$ fulfills a Künneth formula.
\begin{theorem}\label{thm:homologyoffiberbundlecomplex} Assume that the connection $\varphi$ acts as the identity on the homology of $F.$
Then there is an isomorphism.
\begin{align*}
H_{n}(B\otimes_{\varphi} F) \cong\bigoplus_{p+q=n} H_{p}(B)\otimes H_{q}(F).
\end{align*}
\end{theorem}
\begin{proof} By Theorem \ref{thm:smalldoublecomplex} there is an isomorphism
\begin{gather*}
H_{n}(B\otimes_{\varphi} F) \cong\bigoplus_{p+q=n} H_{p}(H_{q}(B_{\bullet}\otimes F_{\bullet},\id\otimes\del),\del_{\varphi}).
\end{gather*}
Since the vertical differential $\id\otimes\del$ acts by the identity on the base~$B$ we get
 $$H_{q}(B_{p}\otimes F_{\bullet},\id\otimes\del)=B_{p}\otimes H_{q}(F).$$
Since by assumption $\varphi$ acts as the identity on the homology of $F,$ the induced horizontal differential 
$$\del_{\varphi}: B_{p}\otimes H_{q}(F)\to B_{p-1}\otimes H_{q}(F)$$
is equal to $\del\otimes \id$ and we obtain
 $$H_{p}(H_{q}(B_{\bullet}\otimes F_{\bullet},\id\otimes\del),\del_{\phi})=H_{p}(B)\otimes H_{q}(F)$$
 which completes the proof.
\end{proof}
We say that the complex $F$ is \emph{augmented} if there is a map $\epsilon: F_0\to \mathbb{F}_2,$ such that $\epsilon\del^F=0.$ For example, if $F=C(Y)$ is the chain complex of a space $Y$ a natural augmentation is induced by the projection $Y\to \{pt\}$ to a point. In this case, there is a natural map of chain complexes $\pi: F\to \mathbb{F}_2$ given by
\begin{center}
\begin{tikzcd}
F_1 \arrow[r] \arrow[d, "0"] & F_0 \arrow[d, "\epsilon"] \\
0 \arrow[r]                  & \mathbb{F}_2             
\end{tikzcd}
\end{center}
which induces projection and restriction maps
\begin{align*}
\pi_*&: B\otimes_\varphi F\to B,\, b\otimes f\mapsto \pi(f)b\text{ and}\\
\pi^*&: B^*\to (B\otimes_\varphi F)^*,\, b\mapsto b\otimes \pi^{tr}(1).
\end{align*}
Hastings et al. in \cite{hastings2020fiber} consider a special situation, in which these maps induce isomorphisms on the first (co-)homology group.
\begin{theorem} \label{thm:projectioninducesisom}
Under the assumptions that
\begin{enumerate}
\item the connection $\varphi$ acts as the identity on the homology of $F,$
\item the augmentation induces an isomorphism $\epsilon: H_0(F)\to \mathbb{F}_2$ and
\item $H_0(B)=0$
\end{enumerate}
the maps $\pi_*$ and $\pi^*$ induce isomorphism
\begin{align*}
\pi_*&: H_1(B\otimes_\varphi F)\to H_1(B)\text{ and}\\
\pi^*&: H^1(B)\to H^1(B\otimes_\varphi F).
\end{align*}
\end{theorem}
\begin{proof}
Follows from Theorem \ref{thm:homologyoffiberbundlecomplex}.
\end{proof}
This is a slight generalization of Lemma 2.5. in Hastings et~al. where a similar result is obtained by  elementary methods.

Hastings et al. focus on a special choice of $B$, $F$ and $\phi$ and estimate the distance of the resulting code. We will not repeat their arguments here, as we will use different techniques to bound the distance of our codes.

\section{Expander Codes}\label{sec:expander}

In this section we introduce certain classical codes and discuss their properties which we will use in our construction of balanced product codes in \Cref{sec:balanced_products,sec:explicit}.
We give a homological definition of Tanner codes which are constructed from graphs and local codes. 
Next, we prove helpful results for expander graphs. 
We then discuss the classical construction of expander codes due to Sipser--Spielman~\cite{sipser1996expander} which combines expander graphs and local codes. We analyze their combinatorial properties in detail as they are a central ingredient for the distance bounds of the balanced product codes constructed in \Cref{sec:explicit}.
As we need expander codes whose symmetries we can control, we introduce a construction of expander graphs due to Lubotzky--Phillips--Sarnak and a construction based on hyperbolic geometry.
We also discuss constructions of potential local codes with good parameters and a slight generalization of the Gilbert--Varshamov bound which guarantees the existence of codes with parameters suitable for the proof of the main theorem in \cref{sec:explicit}.

\subsection{Tanner codes and Local Systems}\label{sec:localsystems}
It was shown by Tanner~\cite{tanner_local_codes} that a large (global) code~$X$ can be improved if its parity checks are replaced by the parity checks of a small (local) code~$L$.\footnote{Tanner calls the local codes ``subcodes'' in~\cite{tanner_local_codes}.}
In order to be able to do this we need that the parity checks of the global code~$X$ are all of weight~$s$ and the local code has block length~$s$.
We will restrict ourselves to the case where the global code~$X$ is a graph code with variables given by edges~$X_{1}$ and checks by vertices~$X_0$.

Let $X$ be a finite $s$-regular graph. For a vertex $v\in X_0,$ denote by $\delta v\subset X_{1}$ the set of incident edges.
Further, let 
\begin{center}
	\begin{tikzcd}
	L = (L_1 \arrow[r,"\del^L"] & L_0)
	\end{tikzcd}
\end{center}
denote a $[s,k,d]$-code called the \emph{local code}. 
Denote by~$\mathcal{B}$ the distinguished basis of $L_1$ given by the $1$-cells (bits) of the code~$L$. 
Furthermore, fix for any $v\in X_0$ a bijection
\begin{align*}
	\Lambda_v: \delta v \to \mathcal{B}
\end{align*}
and let $\Lambda=\{\Lambda_v\}_{v\in X_0}.$ 
Each $\Lambda_v$ is simply a labelling of the edges around the vertex~$v$.

The global code associated to the graph~$X$ and local code~$L$ is given by the complex
\begin{center}
	\begin{tikzcd}
	C(X,L,\Lambda)=(C_1(X) \arrow[r,"\del"] & C_0(X)\otimes L_0)
	\end{tikzcd}
\end{center}
and the differential is defined via
\begin{align}\label{eqn:boundary_local_system}
	\del e = v \otimes \del^L \Lambda_{v}(e) + w \otimes \del^L \Lambda_{w}(e)
\end{align}
where $e\in X_{1}$ is an edge connecting vertices $v,w\in X_0$. We refer to these codes as \emph{Tanner codes}.
We note that in~\cite{tanner_local_codes} Tanner codes are defined in terms of their bipartite Tanner graph. Our definition is a special case of this, by interpreting the boundary operator~$\del$ as a bipartite adjacency matrix.

Also, we mention in passing that it was observed in~\cite{meshulam2018graph} that Tanner codes can also be understood in terms of {\em twisted homology}, where the boundary operator defined in~\Cref{eqn:boundary_local_system} takes values in the $L_0$-valued functions over~$X_0$ instead of the $\mathbb{F}_2$-vector space $C_0(X)\otimes L_0$.
Both definitions are clearly equivalent, but the one given here is more convenient for our purposes.

Not every choice of graph $X$ will give interesting codes.
In the following section we will discuss infinite families of graphs~$\{X^i\}_{i\in \mathbb{N}}$ called expander graphs which combined with a suitable local code~$L$ will yield families of good classical codes. In this case we will refer to the Tanner codes as \emph{expander codes}.
This is an idea due to Sipser--Spielman~\cite{sipser1996expander}.

Let us also note already here that the choice of local labels~$\Lambda$ will not affect any of the bounds proved later.
Hence, we will often simply omit~$\Lambda$ and write~$C(X,L)$, although the exact properties of the global code will in fact depend on~$\Lambda$.

\subsection{Expander Graphs}
Expander graphs can be understood intuitively as graphs which are strongly connected.
Strong connectivity by itself is nothing special as the complete graph is clearly as connected as a graph can be.
However, it is non-trivial that infinite families of graphs exist which are strongly connected despite being $s$-regular, i.e. every vertex has constant degree~$s$.

\subsubsection{Basic definitions}
The connectivity of a graph~$X$ can be quantified by the \emph{Cheeger constant}
\begin{align*}
	h(X) = \min_{ \substack{S\subset V \\ 0<|S|< \frac{|X_0|}{2}} } \frac{|\delta S|}{|S|}
\end{align*}
where $\delta S = \{ \{ u, v \} \in X_{1} \mid u \in S, v \in X_0 - S \}$ is the set of edges connecting~$S$ with its complement.\footnote{We may equivalently think of~$S$ as a $0$-chain in $C_0(X)$, $\delta$ as the coboundary operator and~$|\cdot|$ as the Hamming weight.}
When~$h(X)$ is small it means that we can disconnect a relatively large number of vertices (those in~$S$) from the rest of the graph by removing a relatively small number of edges (those in~$\delta S$).

There are other measures equivalent to the Cheeger constant.
In particular, \emph{spectral expansion} will be useful to us here.
From now on we will assume that~$X$ is a connected, $s$-regular graph.
Let~$A$ be the adjacency matrix of the graph~$X$.
The largest eigenvalue of $A$ will always be~$s$.
Let~$\lambda_2$ be the second largest eigenvalue of $A$.
The \emph{Cheeger inequalities} relate the Cheeger constant with~$\lambda_2$:
\begin{align*}
	\frac{1}{2}(s-\lambda_2) \leq h(X) \leq \sqrt{2 s (s - \lambda_2)}
\end{align*}
The lower the second eigenvalue~$\lambda_2$ the better of an expander the graph~$X$ is.
There exists a slight variation of the lower bound.
\begin{lemma}\label{lem:relative_cheeger}
	Let $X$ be an $s$-regular graph with second largest eigenvalue~$\lambda_2$ and let $S\subset X_0$ with $|S| < \alpha |X_0|$.
	Then it holds that
	\begin{align}
		 (1-\alpha)(s-\lambda_2) \leq  \frac{|\delta S|}{|S|} .
	\end{align}
\end{lemma}
\begin{proof}
The proof is a trivial extension of the standard proof of the Cheeger inequalities.
\end{proof}

Naturally, we might ask what is the largest expansion rate that we can hope for.
This is answered in terms of spectral expansion by the following theorem.
\begin{theorem}[Alon--Boppana bound]\label{thm:alon-boppana}
Let $X$ be an $s$-regular graph with second largest eigenvalue~$\lambda_2$ then
\begin{align*}
\lambda_2 \geq 2\sqrt{s-1} - \frac{2\sqrt{s-1}-1}{\lfloor \operatorname{diam}(X)/2 \rfloor} .
\end{align*}
\end{theorem}
In particular, for a family of $s$-regular graphs~$\{X^i\}$ such that $|X^i| \rightarrow \infty$ for $i\rightarrow \infty$ we have the lower bound $2\sqrt{s-1}$ on all of their second largest eigenvalues.
Graphs which saturate this bound are called \emph{Ramanujan graphs}.

It is a classic result that random graphs are expanders with high probability.
Furthermore, in~\cite{friedman2003relative} it is shown that random graphs are very close to being Ramanujan.
More concretely, let $\epsilon > 0$ and take a random $s$-regular graph~$X$ then the second largest eigenvalue of~$X$ is $\lambda_2 < 2\sqrt{s-1}+\epsilon$ with high probability (depending on $\epsilon$ and $|X|$).
In \Cref{sec:explicit_expanders} we will discuss explicit constructions of expander graphs.

\subsubsection{Properties of expanders}
We will now discuss some results on expander graphs.

\begin{theorem}[Alon--Chung \cite{alon1988explicit}]
	Let $X$ be an $s$-regular graph with second largest eigenvalue~$\lambda_2$ and $S \subset X_0$ of size $|S|=\gamma |X_0|$ with $0<\gamma < 1$.
	Let $X(S)$ be the subgraph of $X$ induced by~$S$.
	Let
	\begin{align}\label{eqn:definition_beta}
	\alpha =   \gamma^2 + \frac{\lambda_2}{s} \gamma (1-\gamma) .
	\end{align}
	It holds that the number of edges of the subgraph $X(S)$ is upper bounded by 
	\begin{align*}
		| X(S)_1 | \leq \alpha |X_{1}| .
	\end{align*}
\end{theorem}
By negation we obtain the following useful corollary.
\begin{corollary}\label{cor:alon_chung}
	Any set of edges $E \subset X_{1}$ of size $\alpha |X_{1}|$ is incident to more than~$\gamma |X_0|$ vertices.
\end{corollary}

We also have an explicit lower bound for the number of vertices incident to a set of edges~$E \subset X_{1}$.
\begin{lemma}[Edge-to-vertex expansion]\label{lem:edge_to_vertex_expansion}
	Let $E \subset X_{1}$ with $|E| \leq \alpha |X_{1}|$.
	Let $\Gamma(E) \subset X_0$ denote the set of vertices incident to the edges in~$E$.
	We then have
	\begin{align*}
	|\Gamma(E)| \geq \beta |E|
	\end{align*}
	where
	\begin{align}\label{eqn:definition_mu}
		\beta = \frac{\sqrt{\lambda_2^2 + 4 s (s-\lambda_2) \alpha} - \lambda_2}{s(s-\lambda_2) \alpha} .
	\end{align}
\end{lemma}
\begin{proof}
	We can solve \Cref{eqn:definition_beta} for $\gamma$ and obtain
	\begin{align*}
	\gamma(\tilde{\alpha}) = \frac{\sqrt{\lambda_2^2 + 4 s (s-\lambda_2) \tilde{\alpha}} - \lambda_2}{2(s-\lambda_2)}.
	\end{align*}
	Since $\gamma(\tilde{\alpha})$ is a convex function we can lower bound it for~$\tilde{\alpha}$ between~0 and~$\alpha$ by a linear function with slope~$\gamma(\alpha)/\alpha$.
	Combined with \Cref{cor:alon_chung} we obtain
	\begin{align*}
	|\Gamma(E)| \geq \gamma |X_0| \geq \frac{2 \gamma(\alpha)}{s \alpha}\, \alpha |X_{1}| = \beta |E|
	\end{align*}
	as claimed.
\end{proof}

The basic property of spectral expanders given in \Cref{lem:relative_cheeger}---sets of vertices $S$ have large boundary with their complement---can be refined.
The following lemma estimates how many vertices in~$S$ have many edges connecting them with the complement of~$S$.
\begin{lemma}\label{lem:relativevertextoedgeexpansion}
Let $S\subset X_0$ be a subset of vertices with $|S|\leq \alpha |X_0|$. 
Let $0\leq b\leq s$ and denote
$$A=\{v\in S \mid |(\delta S)_v|\geq s-b\}$$
where for $v\in S$ we define $(\delta S)_v=\delta S\cap \delta v.$
Let $$\beta= ((b-\lambda_2)-\alpha(s-\lambda_2))b^{-1}$$ then there is a lower bound
 $$|A|\geq \beta |S|.$$
\end{lemma}
\begin{proof}
Let $B=S-A.$ We partition the boundary $\delta S$ accordingly into
$$(\delta S)_{A}=\bigcup_{v\in A} (\delta S)_v \text{ and } (\delta S)_{B}=\bigcup_{v\in B} (\delta S)_v.$$
Since $|(\delta S)_v|\leq s$ for all $v\in S$ and $|(\delta S)_v|< s-b$ for all $v\in B$ we get
\begin{align*}
(\delta S)_{A}&\leq s |A|\text{ and}\\
(\delta S)_{B}&\leq (s-b) |B|.
\end{align*}
These inequalities together with \Cref{lem:relative_cheeger} give
$$s|A|+(s-b) |B|\geq (1-\alpha)(s-\lambda_2) |S|.$$
Since $S$ is the disjoint union of $A$ and $B$ we can cancel the term $|B|$ to obtain
$$|A|\geq \frac{(b-\lambda_2)-\alpha(s-\lambda_2)}{b} |S|=\beta |S|$$
which is what we wanted to show.
\end{proof}

\subsection{Properties of Expander Codes}

\subsubsection{Parameters}
What is the number of encoded bits of the expander code? --- The local code $L$ is defined by~$s-k_L$ linear constraints.
Each vertex of the expander graph~$X$ thus contributes $(1-k_L/s)s$ constraints to the (global) code~$C(X,L)$.
As some constraints may be linearly dependent\footnote{This happens for example when the all-ones vector is a parity check of the local code~$L$.} and using $2\, |X_{1}| = s\, |X_0|$ we obtain the bound
\begin{align}\label{eqn:sipser_spielman_k}
k_{C(X,L)} \geq (2k_L/s-1)\, |X_{1}| .
\end{align}

\begin{theorem}[Sipser--Spielman \cite{sipser1996expander}]
	Let~$s$ be the block length of the local code~$L$ and $d_L$ its distance. Further, let~$\lambda_2$ the second-largest eigenvalue of the $s$-regular expander graph~$X$.
	Then the distance~$d_{C(X,L)}$ of the expander code~$C(X,L)$ is lower bounded as follows:
	\begin{align}\label{eqn:sipser_spielman_d}
		d_{C(X,L)} \geq  \frac{(d_L -  \lambda_2)\,d_L}{(s-\lambda_2)\,s}\, |X_{1}|
	\end{align}
\end{theorem}
\begin{proof} 
	Consider a set of variables (edges) $E \subset X_{1}$ of size
	\begin{align}\label{eqn:sipser_spielman_proof_weight}
	|E| = \frac{s|X_0|}{2} \left( \gamma^2 + \frac{\lambda_2}{s} \gamma (1-\gamma) \right).
	\end{align}
	By \Cref{cor:alon_chung} we have that these edges are incident to more than~$\gamma |X_0|$ vertices.
	Let~$X(E)$ be the subgraph induced by~$E$.
	The average number of edges incident to each vertex in~$X(E)$ is
	\begin{align*}
	\frac{2|E|}{X(E)^0} \leq \frac{2|E|}{\gamma |X_0|} = s\left( \gamma + \frac{\lambda_2}{s}(1-\gamma) \right) .
	\end{align*}
	We can now consider a simple probabilistic argument to constrain the weight of code words:
	If the average number of edges per vertex is smaller than the code distance of the local code~$d_L$ then~$E$ can not be a code word of the global code (the expander code).
	This means we must have
	\begin{align*}
	s\left( \gamma + \frac{\lambda_2}{s}(1-\gamma) \right) < d_L
	\end{align*}
	which is equivalent to
	\begin{align*}
	\gamma < \frac{d_L - \lambda_2}{s - \lambda_2} .
	\end{align*}
	
	The above reasoning implies that for $E$ to be a code word we must have $\gamma \geq (d_L - \lambda_2)/(s - \lambda_2)$ and substituting this into \Cref{eqn:sipser_spielman_proof_weight} gives the result.
\end{proof}
For the bound to be non-trivial we need that the distance of the local code must be strictly larger than the second eigenvalue of~$X$
\begin{align}\label{eqn:expandercode_local_distance_demand}
d_L > \lambda_2.
\end{align}

We note that the bound actually given in~\cite[Lemma 15]{sipser1996expander} is $d_{C(X,L)} \geq n \left( d_L - \lambda_2\right)^2 / \left( s-\lambda_2 \right)^2$ which is slightly worse than \Cref{eqn:sipser_spielman_d}.
It is obtained by dropping the second term in \Cref{eqn:sipser_spielman_proof_weight} which implicitly assumes that \Cref{eqn:expandercode_local_distance_demand} holds.

\subsubsection{Expansion properties of Tanner codes}
A matrix~$A\in\mathbb{F}_2^{m\times n}$ is called {\em $(\alpha,\beta)$-expanding} if for any $x\in \mathbb{F}_2^{n}$ with $|x|\leq \alpha n$ we have that $|Ax| \geq \beta |x|.$
In this section we will analyse the expansion properties of the Tanner code
\begin{center}
	\begin{tikzcd}
	C(X,L,\Lambda)=(C_1(X) \arrow[r,"\del"] & C_0(X)\otimes L_0)
	\end{tikzcd}
\end{center}
given the spectral expansion $\lambda_2$ of the graph $X$ and the distances $d_L, d_{L^\perp}$ of the local code $L$ and its dual~$L^\perp$.
The following theorem shows that small errors will violate a large amount of checks.
\begin{theorem}\label{thm:expanderviolatedchecks}
For $\alpha > 0$ let $x\in C_1(X)$ such that  $|x|\leq\alpha |X_{1}|.$ Then $|\del(x)|\geq\beta|x|$ where $\beta=\beta'\beta''$ and
\begin{align*}
\beta'&=\frac{\sqrt{\lambda_2^2 + 4 s (s-\lambda_2) \alpha} - \lambda_2}{s(s-\lambda_2) \alpha}\\
\beta''&=\frac{(d_L-\lambda_2)-\frac{4\alpha}{s}(s-\lambda_2)}{d_L}
\end{align*}
\end{theorem}
\begin{proof}
Denote by $E\subset X_{1}$ the subset of edges corresponding to the $1$-chain $x.$ Then $|E|=|x|.$ Denote by $S=\Gamma(E)\subset X_0$ the set of vertices incident to $E.$ In the notation of \Cref{lem:relativevertextoedgeexpansion}, let
$$A=\{v\in S \mid |(\delta S)_v|\geq s-d_L\}.$$
Since $\delta S$ and $E$ are disjoint, for all $v\in A$ we have
$$1\leq|\delta v\cap E|<d_{L}.$$
Hence, at least one check at every vertex in $A$ is violated and
$$|\del(x)|\geq |A|.$$
Using the bounds for the edge-to-vertex vertex expansion of~$X$, see \Cref{lem:edge_to_vertex_expansion},
we get
$S\geq \beta'|E|.$
Now $$|S|\leq 2|E|\leq 2\alpha|X_{1}|=\frac{4\alpha}{s}|X_0|$$
Then \Cref{lem:relativevertextoedgeexpansion} yields
$|A|\geq\beta''|S|.$ So, in total, we obtain the desired inequality.
\end{proof}

Dually, the following theorem estimates the number of bits involved in any set of checks.
\begin{theorem}\label{thm:expanderbitdegree}
For $\alpha > 0$ let $y\in C_0(X)\otimes L_0$ such that $|y|\leq \alpha |X_0|(s-k_L).$ Then $|\delta(y)|\geq \beta y$
where 
\begin{align*}
	\beta= \frac{(d_{L^\perp}-\lambda_2)-\alpha(s-k_L)(s-\lambda_2)}{(s-k_L)d_{L^\perp}}
\end{align*}
\end{theorem}
\begin{proof}
Let $S$ be the subset of vertices appearing non-trivially in $y.$ Hence
$$y=\sum_{v\in S} v\otimes c_v$$
for appropriate $0\neq c_v\in L_0.$ Hence, clearly
\begin{equation}\label{eqn:Sandylol}
	(s-k_L)^{-1}|y|\leq |S|\leq |y|.
\end{equation}

In the notation of \Cref{lem:relativevertextoedgeexpansion}, let
$$A=\{v\in S \mid |(\delta S)_v|\geq s-d_{L^\perp}\}.$$
By definition of the differential in the Tanner code, we have $|\delta(v\otimes c_v)|=|\delta_L(c_v)|\geq d_{L^\perp}.$ Hence for each vertex $v\in A,$ there  is at least one edge $e\in(\delta S)_v$ such that $e$ appears in $\delta(v\otimes c_v).$ Since $e\in(\delta S)_v$ it appears in no other term $\delta(v'\otimes c_{v'})$ for $v\neq v'\in S.$ Hence
$$|\delta(y)|\geq |A|.$$
Now the statement follows using \Cref{lem:relativevertextoedgeexpansion} and \Cref{eqn:Sandylol}.
\end{proof}

\subsection{Explicit Constructions of Expander Graphs}\label{sec:explicit_expanders}
In order to be able to apply the construction outlined in \Cref{sec:balanced_products} we need to have full control over the automorphism group of the expander graphs.
This excludes in particular any randomized constructions of expanders.
It turns out that projective linear groups give rise to several families of expanders and we outline two approaches below.

\subsubsection{LPS-expander}
The first construction we discuss is in fact the first explicit construction of Ramanujan expanders, which satisfy the optimal bound $\lambda_2 < 2\sqrt{s-1}$.
It was achieved by Lubotzky, Phillips and Sarnak (LPS) in~\cite{lubotzky1988ramanujan} by defining certain Cayley graphs which we will now discuss.

Consider a group~$G$ with generating set~$S$.
We assume that~$S$ is symmetric, i.e. if $s \in S$ then $s^{-1} \in S$.
The \emph{(undirected) Cayley graph}~$\operatorname{Cay}(G,S)$ consists of the vertex-set~$G$ and any two vertices~$g,g'\in G$ are connected by an edge if and only if there exists an~$s\in S$ such that $g' = sg$.
We stress that (among many other properties) the expansion of Cayley graphs~$\operatorname{Cay}(G,S)$ depends on the choice of the generating set~$S$.

LPS-expanders are Cayley graphs with respect to the groups~$\operatorname{PGL}(2,q)$ or~$\operatorname{PSL}(2,q)$ which we will now define.
Let~$\mathbb{F}_q$ be the finite field of order~$q$ and let~$\operatorname{GL}(2,q)$ be the group of invertible matrices in~$\mathbb{F}_q^{2\times 2}$.
The {\em projective general linear group} $\operatorname{PGL}(2,q)$ is defined as
\begin{align*}
\operatorname{PGL}(2,q) = \operatorname{GL}(2,q) / \mathbb{F}_q^{\times} .
\end{align*}
The {\em special linear group} $\operatorname{SL}(2,q)$ is the kernel of the determinant $\det : \operatorname{GL}(2,q) \to  \mathbb{F}_q^{\times}$.
The {\em projective special linear group} $\operatorname{PSL}(2,q)$ is $\operatorname{SL}(2,q)$ up to signs
\begin{align*}
\operatorname{PSL}(2,q) = \operatorname{SL}(2,q) / \lbrace \pm 1 \rbrace .
\end{align*}

We will now define the generating sets of the LPS-expanders.
Let~$p$ be a prime and let $$\tilde{S}_p = \lbrace (a,b,c,d)\in \mathbb{Z}^4 \mid a^2+b^2+c^2+d^2 = p \rbrace .$$
A classical theorem by Jacobi implies that $|\tilde{S}_p| = 8(p+1)$.
We have two cases:
\begin{enumerate}
\item { $p \equiv 1 \mod 4$}: For any $(a,b,c,d)\in \tilde{S}_p$ we have that exactly one element is odd and the rest are even.
Define $S_p \subset \tilde{S}_p$ be the subset where $a$ is positive and odd.
\item { $p \equiv 3 \mod 4$}: For any $(a,b,c,d)\in \tilde{S}_p$ we have that exactly one element of $\{a,b,c,d\}$ is even and the rest are odd.
Define $S_p \subset \tilde{S}_p$ be the subset where $a$ is even and the first non-zero element is positive.
\end{enumerate}
It can be shown that either way~$|S_p|=p+1$.
Now let~$q$ be prime and $x,y\in \mathbb{F}_q$ such that $x^2+y^2+1=0.$ We define the following set of matrices:
\begin{align*}
	S_{p,q} = \left\lbrace \begin{bmatrix}
	a + bx + d y & -b y + c + dx \\
	-by-c+dx & a-bx-dy
	\end{bmatrix} \mid (a,b,c,d) \in S_{p} \right\rbrace
\end{align*}
where the product between integers and elements of~$\mathbb{F}_q$ is defined in the obvious way.

For $a\in \mathbb{Z}$ and~$p$ an odd prime we define the {\em Legendre symbol} 
\begin{align*}
\left(\dfrac{a}{p}\right) = 
\begin{cases}
0, &\text{ if $p \mid a$}\\
1, &\text{ if $p \nmid a$ and $a$ is a square modulo $p$}\\
-1, &\text{ if $p \nmid a$ and $a$ is not a square modulo $p$.}
\end{cases}
\end{align*}

\begin{theorem}[Lubotzky--Phillips--Sarnak \cite{lubotzky1988ramanujan}]
	Assume that~$p$ and~$q$ are distinct, odd primes such that $q > 2 \sqrt{p}$.
	If $$\left(\dfrac{p}{q}\right) = 1$$ then the set~$S_{p,q}$ is a symmetric generating set of $\operatorname{PSL}(2,q)$ and we define $X_{p,q} = \operatorname{Cay}(\operatorname{PSL}(2,q),S_{p,q})$.
	
	Otherwise we must have $$\left(\dfrac{p}{q}\right) = -1$$ and the set~$S_{p,q}$ is a symmetric generating set of $\operatorname{PGL}(2,q)$ and we define $X_{p,q} = \operatorname{Cay}(\operatorname{PGL}(2,q),S_{p,q})$.
	
	In both cases the graphs $X_{p,q}$ are connected, $p+1$--regular expanders with $\lambda_2 < 2 \sqrt{p}$.
\end{theorem}
For more information on LPS-expanders see~\cite{davidoff2003elementary}.

\subsubsection{Hyperbolic tessellations}
Expander graphs can also be generated by purely geometric means, namely by considering regular tessellations of hyperbolic space.
In particular, there exist infinitely many regular tessellations of the hyperbolic plane~$\mathbb{H}^2$.
A regular tessellation is an edge-to-edge covering of the plane by regular polygons.
They are completely specified by the number of sides of the polygon~$r$ and the number of polygons meeting at each vertex~$s$.
This data is captured in the \emph{Schl\"afli symbol}~$\{r,s\}$.

It is possible, using the symmetry group~$G_{r,s}$ of a regular tessellation~$\{r,s\}$, to construct families of closed hyperbolic surfaces of increasing size.
From this we can obtain infinite families of graphs~$\{X^i\}$ by ignoring the faces or, in mathematical terms, by taking their the 1-skeleton~$(X^i)_{\leq 1}$.
The hyperbolic surfaces are constructed by factoring out normal subgroups~$N$ of the symmetry group~$G_{r,s}$.
The details of this procedure are outlined in Appendix~\ref{sec:hyperbolic}.

It was proven in~\cite{conder2020constructing} (using results from \cite{salehi2017super}) that all hyperbolic regular tessellations $\{r,s\}$ give families of $s$-regular expander graphs.
Unfortunately, we are not aware of any explicit bounds on~$\lambda_2$.

\begin{figure}
	\centering		
	\includegraphics[angle=180,origin=c,width=0.7\linewidth]{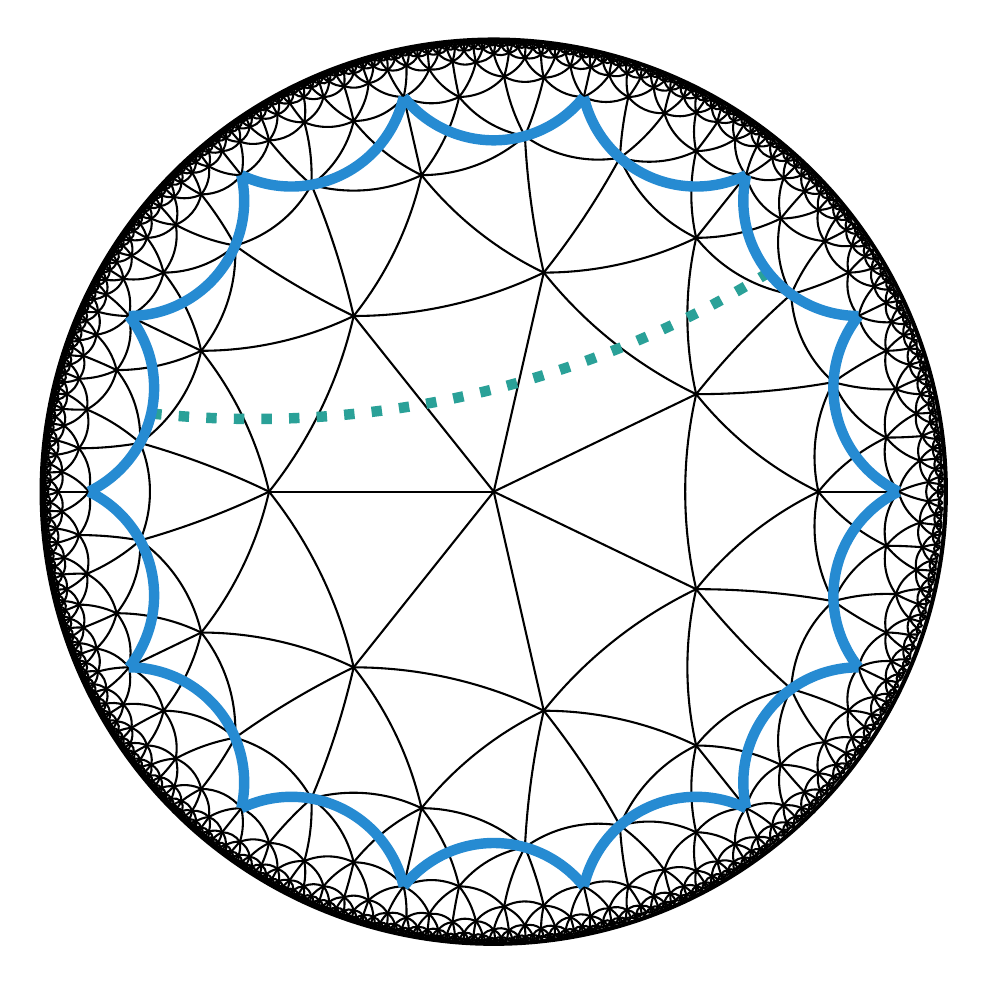}
	\caption{A hyperbolic genus 3 surface supporting the $\{3,7\}$ tessellation. It is called the Klein quartic.
		The dotted line indicates the translation~$\tau = (\rho\, \sigma^{-2})^4$ under which points are being identified.
		The symmetry group of the surface is $\operatorname{Isom}(\mathbb{H}^2/N_\tau)=\operatorname{PSL}(2,7)\rtimes \mathbb{Z}_2$.
	}
	\label{fig:klein_quartic}
\end{figure}

One method of obtaining normal subgroups is by taking a matrix representation of~$G_{r,s}$ and reducing the matrix coefficients modulo a prime number.
This induces a group homomorphism~$\pi_p$.
Since the kernel of any group homomorphism is normal we can consider $N_p = \ker \pi_p$ as candidates.
In the context of quantum codes this procedure was used in~\cite{guth2014quantum,breuckmann2020single} in order to construct codes based on compact four-dimensional hyperbolic manifolds.

The group of orientation preserving symmetries of a regular tessellation $G_{r,s}^+ = \left\langle \rho, \sigma \mid \rho^{r}, \sigma^s, (\rho \sigma)^2 \right\rangle$ is generated by the order-$r$ rotation around a face and the order-$s$ rotation around a vertex.
It naturally embeds into the real projective linear group (cf.~\Cref{eqn:triangle_group_in_PSL} in Appendix~\ref{sec:hyperbolic}).
One representation of $G_{r,s}^+$ in $\operatorname{SL}(2,\mathbb{R})$ is due to Magnus~\cite{magnus1974noneuclidean}.
It is obtained by mapping
\begin{align*}
\begin{split}
\rho &\mapsto R = \frac{i}{\sin(\pi/s)} \begin{bmatrix}
\cos(\pi/r)\, e^{i\pi/s} &  \xi\, e^{-i\pi/s} \\
- \xi\, e^{i\pi/s} & - \cos(\pi/r)\, e^{-i\pi/s}
\end{bmatrix}
\end{split}
\end{align*}
where $\xi = \sqrt{\cos(\pi/r)^2 - \sin(\pi/s)^2}$ and
\begin{align*}
\sigma &\mapsto S = \begin{bmatrix}
e^{i\pi/s}  & 0 \\
0 &  e^{-i\pi/s}
\end{bmatrix} .
\end{align*}
The corresponding elements $\bar{R}$ and $\bar{S}$ in $\operatorname{PSL}(2,\mathbb{R})$ fulfill exactly the same relations as~$\rho$ and~$\sigma$ so that $G_{r,s}^+ \simeq \langle \bar{R}, \bar{S}  \rangle$.

We now observe that for certain choices of~$\{r,s\}$ the entries of~$R$ and~$S$ are algebraic integers, i.e. they are roots of normed polynomials with integer coefficients.\footnote{For example, for the Magnus representation this is the case whenever the ratio~$\cos(\pi/r)/\sin(\pi/s)$ is an algebraic integer.}
In that case we can adjoin the missing roots to $\mathbb{Z}$ giving a ring~$A$ and obtain $G_{r,s}^+ \leq \operatorname{PSL}(2,A)$.
Reducing~$A$ modulo a prime~$p$ via $\pi_p$ maps the matrix coefficients into a finite field~$\mathbb{F}_{p^m}$.
Hence, we obtain a subgroup of~$\operatorname{PSL}(2,p^m)$.

It turns out that there are cases where this mapping~$\pi_p : G_{r,s}^+ \rightarrow \operatorname{PSL}(2,p^\nu)$ is an epimorphism so that in fact~$\operatorname{im} \pi_p \simeq \operatorname{PSL}(2,p^\nu)$.
Additionally, in order for~$\operatorname{PSL}(2,p^\nu)$ to be the symmetry group of a compactified surface we need that the corresponding $\ker \pi_p$ is torsion-free.
The epimorphic images of~$G_{r,s}^+$ onto~$\operatorname{PSL}(2,q)$ with torsion-free kernel are classified by the following theorem.\footnote{The results of Langer--Rosenberger in~\cite{langer1989erzeugende} are more general, we only quote the part which is relevant for us here.}
It is in fact more general, as it also covers cases where there are epimorphic images of~$G_{r,s}^+$ onto~$\operatorname{PGL}(2,q)$ with torsion-free kernel.

We define~$\mu(p,n)$ to be the smallest~$\nu \in \mathbb{N}$ such that
\begin{align*}
p^\nu \equiv \pm 1 \begin{cases}
\operatorname{mod} n, &\text{if } 2 \nmid n \\
\operatorname{mod} 2n, &\text{if } 2 \mid n .
\end{cases} 
\end{align*}
We extend this definition to any sequence of integers~$n_1,\dotsc,n_w\in \mathbb{N}$ either equal or coprime to~$p$ by defining~$\mu(p,n_1,\dotsc,n_w)$ as the least common multiple of~$\mu(p,n_1),\dotsc,\mu(p,n_w)$.
For some tuple $(n',m',l')\in \mathbb{N}^3$ define the following two conditions:
\begin{itemize}
\item[(A)] $m'=2$, $l'$ and $\nu = \mu(p,n',l')$ are even, $\mu(p,m')$ and $\mu(p,l')$ do not divide $\nu /2$, $\mu(p,n')$ divides $\nu /2$ and $p^{\frac{\nu}{2}} \equiv \pm 1 \mod l'$.
\item[(B)] $m'$, $l'$ and $\nu = \mu(p,n',m',l')$ are even, $\mu(p,m')$ and $\mu(p,l')$ do not divide $\nu/2$, $\mu(p,n')$ divides $\nu /2$, $p^{\frac{\nu}{2}} \equiv \pm 1 \mod m'$ and $p^{\frac{\nu}{2}} \equiv \pm 1 \mod l'$.
\end{itemize}
\begin{theorem}[Langer--Rosenberger~\cite{langer1989erzeugende}]
	Let $\{r,s\}$ be the Schl\"afli symbol of a hyperbolic tessellation and~$p \geq 3$ a prime with either a)  $p \nmid r s$, or b) $r=p$ and $p \nmid s$, or c) $s=p$ and~$p \nmid r$.
	
	The group $\operatorname{PSL}(2,p^{\nu})$, where $\nu = \mu(p,2,r,s)$, is the epimorphic image of~$G_{r,s}^+$ with torsion-free kernel if either
	\begin{enumerate}
		\item $r$ and $s$ are odd, or
		\item $r$ or $s$ (or both) are even and there exists no permutation $(n',m',l')$ of $(2,r,s)$ such that condition (A) or (B) are fulfilled.
	\end{enumerate}
	
	The group $\operatorname{PGL}(2,p^{\nu})$, where $\nu = \mu(p,2,r,s)/2$, is the epimorphic image of~$G_{r,s}^+$ with torsion-free kernel if $r$ or $s$ (or both) are even and there exists a permutation $(n',m',l')$ of $(2,r,s)$ such that condition~(A) or~(B) is fulfilled.
\end{theorem}

A famous example is the \emph{Klein quartic} shown in \Cref{fig:klein_quartic}.
It is a hyperbolic surface of genus~3 which is tessellated by~$\{3,7\}$ and whose orientation-preserving symmetries form the group~$\operatorname{PSL}(2,7)$.
The kernel of the reduction is isomorphic to $N^+_\tau \leq G_{3,7}^+$ which is generated by all conjugates of the translation $\tau = (\rho\, \sigma^{-2})^4$ (cf. \Cref{eqn:compactification_translation}).

The discussion above indicates that hyperbolic expanders are related to Cayley graphs $\operatorname{Cay}\left(\operatorname{PSL(2,q)},\lbrace \bar{R},\bar{S},\bar{R}^{-1}, \bar{S}^{-1} \rbrace\right)$.
This relation is made precise in~\cite{conder2020constructing} by defining a quasi-isometry between them.

\begin{figure}
	\centering
	\includegraphics[width=0.85\linewidth]{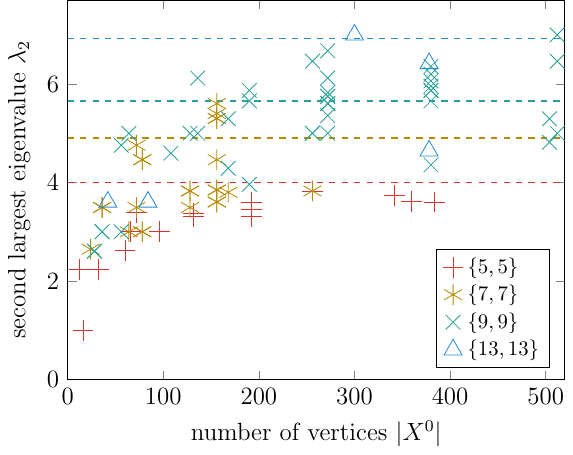}
	\caption{The second largest eigenvalues of hyperbolic surfaces of small size labeled by their Schl\"afli symbol $\{r,s\}$. The dashed lines indicate the asymptotic lower bound~$2\sqrt{s-1}$ of \Cref{thm:alon-boppana}.
	}
	\label{fig:hyperbolic_spectrum}
\end{figure}

\subsection{Local Codes}
In order to construct good expander codes it suffices for the local codes to satisfy the constraints given by \Cref{eqn:sipser_spielman_k} and \Cref{eqn:expandercode_local_distance_demand}, namely that it has to have rate $k_L/2 > 1/2$ and distance larger than the second eigenvalue of the graph $d_L > \lambda_2$.

\subsubsection{Goppa Codes}
A large family of codes which provides suitable candidates for the local codes~$L$ are \emph{Goppa codes}.
We consider their binary version here:
Let $m$ be a positive integer and let $g\in \mathbb{F}_{2^m}[x]$ be a polynomial of degree~$t$.
Furthermore, let $\gamma_1,\dotsc,\gamma_s\in \mathbb{F}_{2^m}$ be chosen such that $g(\gamma_i)\neq 0$ for all~$i$.
A vector $c = (c_1,\dotsc,c_s) \in \mathbb{F}_2^s$ is a code word of the Goppa code~$\operatorname{GC}(g,\{\gamma_i\})$ if and only if
\begin{align*}
	\sum_{i=1}^{s} \frac{c_i}{x-\gamma_i} \equiv 0 \mod g .
\end{align*}
Note that we restricted the code to be over~$\mathbb{F}_2$ which naturally embeds into~$\mathbb{F}_{2^m}$.
It can be shown that these codes come with the following bounds.
\begin{theorem}
The (binary) Goppa code~$\operatorname{GC}(g,\{\gamma_i\})$ defined above encodes~$k_{\operatorname{GC}} \geq s-mt$ many bits and has distance~$d_{\operatorname{GC}} \geq 2t+1$.
\end{theorem}
Hence, in order to satisfy the constraints \Cref{eqn:sipser_spielman_k} and \Cref{eqn:expandercode_local_distance_demand} we only have to find a suitable polynomial~$g$ and elements $\gamma_1,\dotsc,\gamma_s\in \mathbb{F}_{2^m}$.

It turns out that, using some deep results on exponential sums, it is possible to also bound the distance of the dual code of a Goppa code.
\begin{theorem}[Moreno--Moreno~\cite{moreno1991exponential}]
	Assume the polynomial~$g$ has~degree~$t$ and all roots of~$g$ are distinct.
	Let $Z$ be the set of roots of~$g$ and let $\{\gamma_i\}$ be the set $\mathbb{F}_{2^m} - Z$.
	 Then the distance~$d_{\operatorname{GC}^\perp}$ of the dual code $\operatorname{GC}(g,\{\gamma_i\})^\perp$ is bounded from below by
	\begin{align*}
		d_{\operatorname{GC}^\perp} \geq 2^{m-1} - \frac{|Z|-1}{2} - (t-1)\, 2^{m/2} .
	\end{align*}
\end{theorem}

\subsubsection{Cyclic codes and canonical labelings}
To define the expander code~$C(X,L,\Lambda)$ we need to fix a labeling~$\Lambda$ of the edges around each vertex.
As the bounds do not depend on this labeling it can be chosen arbitrarily.
However, it would be desirable if the labeling is compatible with some symmetries of the graph around each vertex.

In \Cref{fig:basis_expander_code} we show an example of an expander code~$C(X,L)$ where~$X$ is obtained from a~$\{3,7\}$ tessellation of a genus-3 surface and the local code~$L$ is the~$[7,4,3]$ Hamming code which is cyclic.
Here there is a canonical choice for the labeling which respects the rotational symmetry around each vertex and of the cyclic symmetry of the local code.
With this labeling the automorphism group acts transitively on the checks.
Hence all checks are given by the orbit of a single check.
Such codes are called {\em single-orbit symmetric codes}, a notion first introduced by Kaufman--Wigderson in~\cite{kaufman2010symmetric}.

It would be desirable if we could construct a family of codes which has this property.
This is achieved by Kaufman--Lubotzky in the context of edge-transitive Ramanujan graphs~\cite{kaufman_lubotzky}.
In this work, it is shown that using the construction of LPS-expanders they can achieve a cyclic symmetry around each vertex.
Moreover they construct suitable cyclic codes with good properties and show that the resulting Tanner codes are single-orbit symmetry codes with linear number of encoded bits and linear distance.

A convenient choice are {\em BCH codes} which are a subclass of cyclic Goppa codes.
The polynomial~$g$ is chosen to be~$x^{t}$ and the $\gamma_i$ are the $i$-th powers  of a primitive $s$-th root of unity, see \cite[Example~8.2.6]{van2012introduction}.

However, there is a subtlety when combining BCH codes with LPS-expanders, namely that the degree of LPS-expanders is even but the number of bits in BCH codes is odd.
In order to overcome this the authors of~\cite{kaufman_lubotzky} apply the doubling process of~\cite{van1991repeated} to appropriately chosen BCH codes.
For hyperbolic expanders this problem disappears as there is no restriction on the degree and thus there is more flexibility in the choice of the cyclic code.

\begin{figure}
	\centering
	\begin{tabular}{ccc}
		\includegraphics[width=0.25\columnwidth]{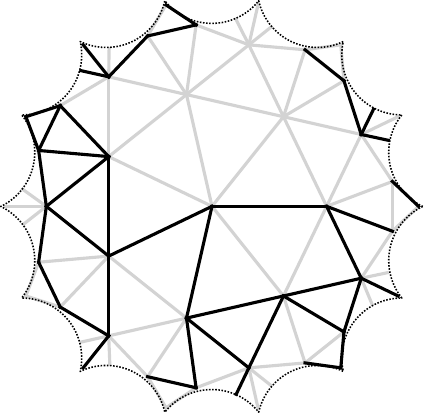} &   \includegraphics[width=0.25\columnwidth]{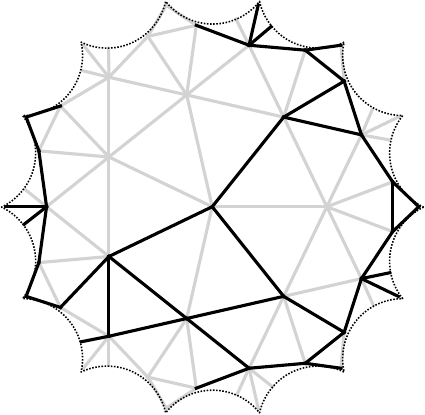} &   \includegraphics[width=0.25\columnwidth]{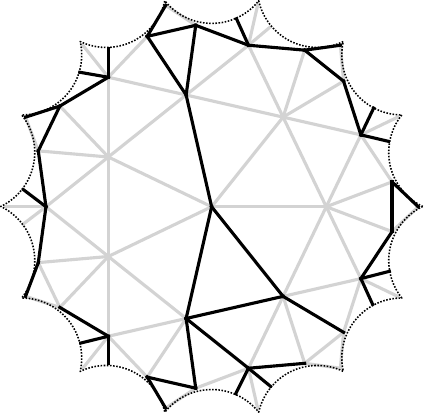} \\
		\includegraphics[width=0.25\columnwidth]{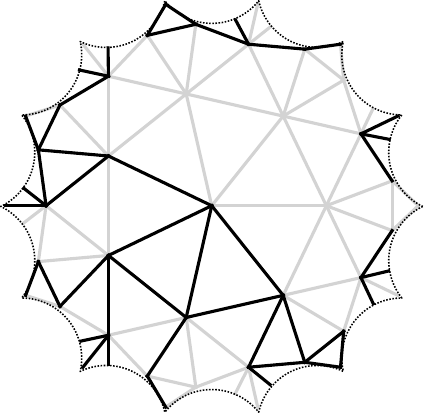} &   \includegraphics[width=0.25\columnwidth]{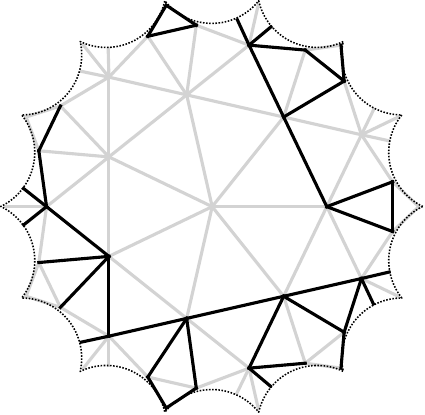} &   \includegraphics[width=0.25\columnwidth]{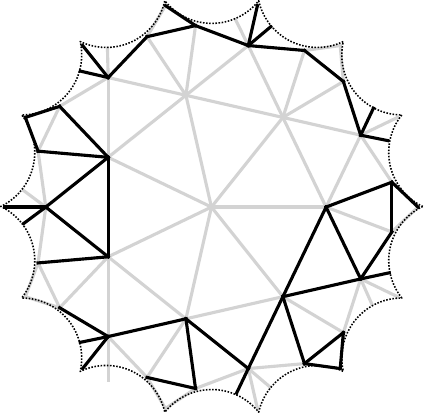} \\
		\includegraphics[width=0.25\columnwidth]{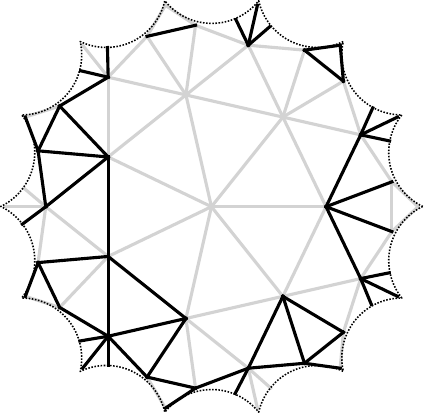} &   \includegraphics[width=0.25\columnwidth]{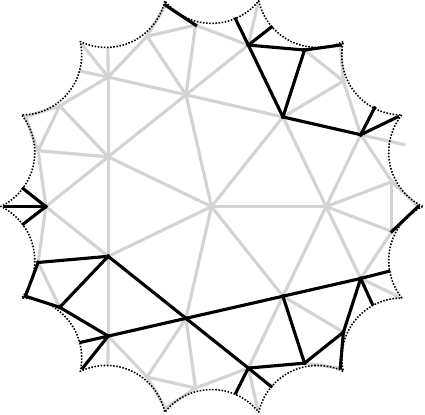} &   \includegraphics[width=0.25\columnwidth]{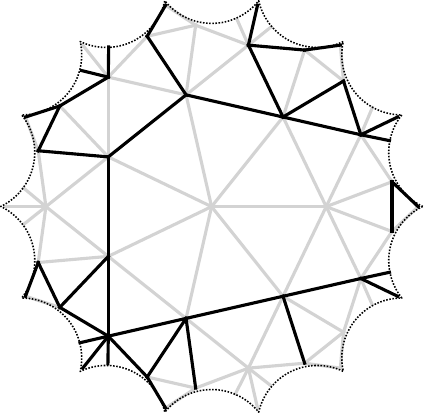} \\
		\includegraphics[width=0.25\columnwidth]{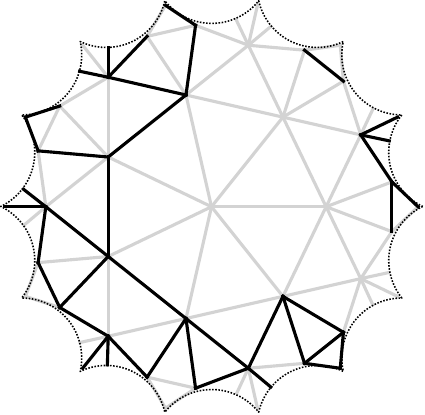} &   \includegraphics[width=0.25\columnwidth]{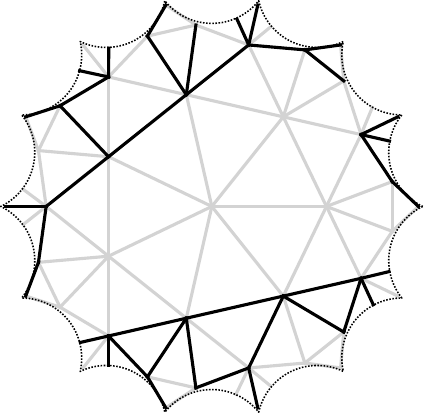} &   \includegraphics[width=0.25\columnwidth]{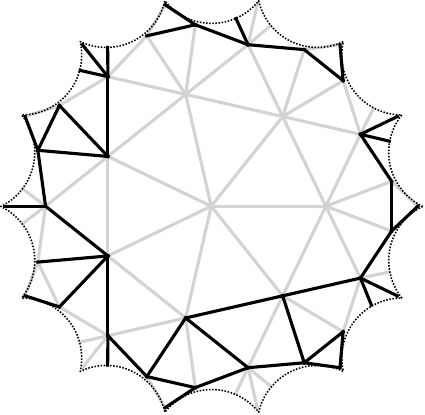} 
	\end{tabular}
	\caption{Basis of a $[84,12,19]$ hyperbolic expander code~$C(X,L)$.
		The graph~$X$ is derived from the~$\{3,7\}$ tessellation of the Klein quartic.
		The local code~$L$ is the $[7,4,3]$ Hamming code which is cyclic.
		Edges correspond to variables which are colored in black if their value is 1 and grey if their value is 0. Around each vertex is a code word of the Hamming code.}
	\label{fig:basis_expander_code}
\end{figure}

\subsubsection{Existence of good local codes}
We can give stronger bounds on the properties of the local codes if we do not have to give an explicit construction.
The following result, which extends the Gilbert-Varshamov bound by a bound for the distance of the dual code, was given in~\cite{panteleev_almostlinear} without a formal proof.
Here we spell out the argument in full detail.
Let $$H_2(\delta) = -\delta \log_2(\delta) - (1-\delta) \log_2 (1-\delta)$$ be the binary entropy function.\footnote{The binary entropy function~$H_2$ and the variable~$\delta$ should not be confused with the second homology class and the coboundary operator. We will only use the former in the context of \Cref{thm:gilbertvarshamovplus}.}
\begin{theorem}[Gilbert--Varshamov+]\label{thm:gilbertvarshamovplus}
	Choose $\delta \in (0,0.11)$. 
	Then for all $n>2/(1/2 - H_2(\delta))$ there exists a binary, linear code~$C$ with number of encoded bits $k_C > n/2$, distance $d_C \geq \delta n$ and dual code~$C^\perp$ with distance $d_{C^\perp} \geq \delta n$.
\end{theorem}
\begin{proof}
	Fix $\delta \in (0,0.11)$.
	Let~$A$ be the fraction of matrices~$G$ in~$\mathbb{F}_2^{k\times n}$ such that there exists an $m\in \mathbb{F}_2^k-\{0\}$ with $|m G| < \delta n$.
	For fixed~$m$ let~$A_m$ be the fraction of matrices~$G$ in~$\mathbb{F}_2^{k\times n}$ such that $|m G| < \delta n$.
	Clearly, we have $$A \leq \sum_{m\in \mathbb{F}_2^k-\{0\}} A_m .$$
	Equivalently,~$A_m$ can be defined as 
	$$\sum_{\substack{y\in \mathbb{F}_2^n\\|y|\leq \delta n - 1}} A_{m,y}$$
	where $A_{m,y}$ of all fractions of matrices~$G$ in~$\mathbb{F}_2^{k\times n}$ with the property that $m G = y.$
	Clearly $A_{m,y}=2^{-n}$ independently of~$m$ and~$y$.
	Hence
	$$A_{m}=  \frac{|\operatorname{B}_0(\delta n -1)|}{2^n}$$
where  $\operatorname{B}_0(\delta n -1)$ denotes the set of all $y\in \mathbb{F}_2^n$ with $|y|\leq \delta n - 1.$ 	We can bound $$|\operatorname{B}_0(\delta n -1)| \leq 2^{H_2(\delta) n}$$ and obtain
	$$A \leq 2^{k-(1-H_2(\delta))n} .$$
	
	Now, let us consider the fraction~$A'$ of \emph{full-rank} matrices in $G\in\mathbb{F}_2^{k\times n}$ such that there exists an $m\in \mathbb{F}_2^k-\{0\}$ with $|m G| < \delta n$.
	Since any matrix which is not of full rank contributes to~$A$, we have $$A' \leq A \leq 2^{k-(1-H_2(\delta))n}.$$
	The group of invertible $\mathbb{F}_2^{k\times k}$-matrices~$\operatorname{GL}(k,2)$ operates {\em freely} from the left on the matrices~$G$ of full rank, so that all of the orbits have the same size.
	The orbits actually correspond to the different $k$-dimensional subspaces of $\mathbb{F}_2^{n}$, so the set of orbits is in fact the set of all codes of block length~$n$.
	In particular, each representative of an orbit has the same distance.
	Hence, the fraction of codes with distance smaller than~$\delta n$ is equal to~$A'$.
	
	Similarly, denote by $B'$ the fraction of subspaces of dimension $n-k$ of maximum distance smaller than $\delta n$.
	By the same arguments as before we can show that $$B' \leq 2^{(n-k)-(1-H_2(\delta))n}.$$	
	
	Since duality $C\mapsto C^\perp$ is a bijection between codes with~$k$, respectively $n-k$ encoded bits~$B'$ is also the fraction of codes $C$ with $k$ encoded bits for which $d_{C^{\perp}} < \delta n.$
	Now, let $D'$ be the fraction of codes $C$ with $k$ encoded bits such that their distance $d_{C} < \delta n$ or their dual distance $d_{C^{\perp}} < \delta n.$ Then $D' \leq A' + B'$.
	We are interested in the case $k>n/2$ which implies $k > n-k$ so that $D' \leq 2^{k-(1-H_2(\delta))n +1}$.
	If the right hand side is smaller than 1 we are done since then we know there must exist a code with the desired properties.
	This condition is equivalent to $k - (1-H_2(\delta))n+1 < 0$.
	Since we are looking for $k$ to be an integer between $n/2$ and $(1-H_2(\delta))n-1$, such $k$ exists if $(1/2 - H_2(\delta))n > 2.$ 		Now a direct calculation shows that $1/2 > H_2(\delta)$ if $\delta \in (0,0.11)$ and the statement follows.
\end{proof}

\section{Quantum Codes from Balanced Products}\label{sec:balanced_products}

In this section we introduce the general framework of balanced product codes.
The balanced product is a notion from topology that out of two spaces, which share a common symmetry, constructs a new space.
We sketch this topological construction and explain its relation to fiber bundles. We then explain how balanced products can be discretized in the setting of graphs. By doing so we introduce several concepts, such as ``trivializations'' and ``twist functions'' in the setting of graphs.
Having done this, we introduce a notion of balanced product of vector spaces and chain complexes, thus yielding a general definition of balanced product codes. Here, we derive a Künneth formula and compare balanced products to  Hastings--Haah--O'Donnell's fiber bundle codes~\cite{hastings2020fiber} and lifted product codes of Panteleev--Kalachev~\cite{panteleev_almostlinear}. 
We discuss balanced products of expander codes and (cyclic) graphs in detail. 
We then introduce a subsystem version of balanced product codes. This is necessary to control certain linear dependencies of checks in expander codes which could potentially lead to unwanted logical operators of low weight.
At the end of the section, we adapt Panteleev--Kalachev's distance theorems for lifted product codes~\cite{panteleev_almostlinear} to our setting.

\subsection{Motivation from Topology}\label{sec:motivationfromtopology}
Our construction of codes is strongly motivated by principal bundles and associated fiber bundles, which we briefly recap now. Let $X$ be a topological space with a \emph{free} right action of a group $H.$ 
Assuming some technical conditions, the quotient map 
$$\pi: X\to X/H$$ is a principal $H$-bundle and in particular a fiber bundle with base $X/H$ and fiber $H.$

\subsubsection{Balanced Products} 
From the bundle~$(X,\pi)$ and any space $Y$ with a left action of $H$ one can construct a new fiber bundle over $X/H$ and with fiber $Y$ in the following way. 
Let~$H$ act anti-diagonally on the Cartesian product~$X\times Y$ via
\begin{align*}
	 (x,y) \cdot h = (x\cdot h,\, h^{-1}\cdot y).
\end{align*} 
The \emph{balanced product} is the set of equivalence classes ($H$-orbits) under the action above:
\begin{align*}
	X \times_{H} Y = (X\times Y)/H
\end{align*}
We denote the elements of $X \times_{H} Y$ with square brackets~$[x,y]$ to highlight that they are not tuples but equivalence classes under the action of~$H$.
Note that we can move a group element $h\in H$ from one side to the other $[x\cdot h, y] = [x, h\cdot y]$.

By projection we obtain a map
$$\pi_{Y}: X \times_{H} Y\to X/H, [x,y]\mapsto xH$$
which is the desired fiber bundle with fiber $Y.$
\subsubsection{Example: Klein Bottle} 

For example, let $X=S^{1}$ be a circle and $H=\mathbb{Z}_{2}$ act via rotation by half a turn. Then $X/H=S^{1}$ is also a circle and $\pi: X\to X/H$ a $2$-fold covering. Let $H$ act on the circle $Y=S^{1}$ by the antipodal map. 
Then the associated balanced product
$S^{1}\times_{\mathbb{Z}_{2}}S^{1}$
is a \emph{Klein bottle} which is a fiber bundle over the circle
$$\pi_{S^1}:S^{1}\times_{\mathbb{Z}_{2}}S^{1}\to S^{1}$$
with fiber $S^{1}.$

\subsubsection{Trivialisations of Bundles}
There is another, coordinatized, perspective on principal bundles and associated bundles which resembles the construction of fiber bundle codes by Hastings--Haah--O'Donnell more closely, see \Cref{sec:fiberbundlecodes}.

For this, choose a \emph{trivialisation} of the bundle $\pi: X\to X/H,$ that is, a covering $\{U_{i}\}$ of $X/H$ by open subsets and homeomorphisms  
$$\psi_{i}: \pi^{-1}(U_{i})\to U_{i}\times H$$
which are compatible with the action of $H.$
If two sets~$U_{i}$ and~$U_{j}$ intersect non-trivially, the trivialisation induces transition maps
$$\varphi_{i,j}: U_{i}\cap U_{j}\to H.$$
These are the continuous version of the \emph{connection} or \emph{twist}~$\varphi$ used in definition of fiber bundle codes by Hastings--Haah--O'Donnell (see \Cref{sec:fiberbundlecodes}).

The bundle $\pi: X\to X/H$ can be completely recovered from the data $\{U_{i}, \varphi_{i,j}\}$ by gluing the spaces $U_{i}\times H$ along the intersections $(U_{i}\cap U_{j})\times H$ using the transition maps $\varphi_{i,j}.$

Similarly, if $Y$ is a space with a left $H$-action, the data $\{U_{i}, \varphi_{i,j}\}$ allows to construct a fiber bundle with fiber~$Y$ by  gluing the spaces~$U_{i}\times Y$ using the action of~$\varphi_{i,j}(x)$ on~$Y$.
In fact, this construction agrees with the balanced product $X\times_{H}Y$ described above.

\subsection{Balanced Product and Bundles on Graphs}\label{sec:balancedproductsgraphs}
We now turn to a discrete version of the topological motivation from the last section.
For simplicity, we only consider one-dimensional cell complexes, i.e. graphs. 
However, the same ideas apply for cell complexes of arbitrary dimension.

\subsubsection{Quotient graph and trivialisation}
Let $X$ be a graph with a free action of a finite group $H$, that is, $H$ acts freely on vertices and edges. 
Further, we assume that there is no edge connecting two vertices of the form $v$ and $vh.$ 
We refer to this as the {\em quotient condition} (see \cite{breuckmann2020single}).
For convenience, we choose an $H$-invariant orientation of the graph $X$ and write $e=(v,v')$ for an oriented edge from $v$ to $v'.$ Such a choice is possible since $H$ acts freely on the edges of $X.$ 

The \emph{quotient graph} $X/H$ is obtained in the following way.
The vertices in $X/H$ are the $H$-orbits $\mathcal{V}$ of vertices in~$X$
$$(X/H)_{0}=\{\mathcal{V}\}.$$
The edges in $X/H$ are the $H$-orbits of $\mathcal{E}$ of edges in~$X$
$$(X/H)_{1}=\{\mathcal{E}\}.$$
Hence, two vertices $\mathcal{V}$, $\mathcal{V}'$ are adjacent in~$X/H$ if there are $v\in \mathcal{V}$ and $v'\in \mathcal{V}$ such that $v$ and $v'$ are adjacent in $X.$  
The orientation on $X$ induces an orientation of~$X/H$.

By the assumptions on the action of $H$ on $X$, the quotient map $$\pi: X\to X/H,\, v\mapsto vH$$
is a principal bundle and in particular an $|H|$-fold covering of~$X/H.$ If $X$ is $s$-regular, so is $X/H.$

 A \emph{trivialisation} of $\pi: X\to X/H$ is a choice of representative $R=\{v\in \mathcal{V}\}$ for each orbit. 
It induces a \emph{connection}~$\phi_{R}$
$$\phi_{R}: (X/H)_{1}\to H, \mathcal{E}=(\mathcal{V},\mathcal{V}')\mapsto \phi_{R}(\mathcal{E}).$$ 

Here, $\phi_{R}(\mathcal{E})$  is defined as follows. Denote by $v,v'\in R$ the representatives of $\mathcal{V},\mathcal{V}'.$ Then $\phi_{R}(\mathcal{E})$ denotes the unique element in $H$ such that $(v,v'\phi_{R}(\mathcal{E}))$ is an edge in $X.$
\begin{center}
\begin{tikzcd}
vh^{-1} \arrow[rd]    &                                               \\
v \arrow[rd]          & v' \arrow[d, "\phi_R(\mathcal{E})=h", dotted, bend left] \\
vh \arrow[rd]         & v'h                                            \\
                      & v'h^2                                          \\
\mathcal{V} \arrow[r,"\mathcal{E}"] & \mathcal{V}'                                  
\end{tikzcd}
\end{center}

By choosing an appropriate trivialisation $R$ it is always possible to \emph{locally} simplify $\phi_{R}$ such that, for example, $\phi_{R}(\mathcal{E})=1$ for all edges incident to a single fixed vertex. However, if $X$ is connected, it is not possible to find a trivialisation $R$ such that \emph{globally} $\phi_{R}=1$.

It is straightforward to see that one can recover the bundle $\pi:X\to X/H$ from $X/H$ and $\phi_{R}$ and we will make use of both perspectives.

\subsubsection{Balanced Products of Graphs}
From the fiber bundle $\pi: X\to X/H$ and another graph $Y$ with a left action of $H$ we can form the balanced product
$X\times_{H} Y=(X\times Y)/H$
which is the quotient of the Cartesian product complex $X\times Y$ by the anti-diagonal action of $H.$
The map $$\pi_{Y}:X\times_{H} Y\to X/H$$ has fibers $Y$ and can be reconstructed from $X/H$, the transition maps $\phi_{R}$ and $Y.$ 
An example is shown in \Cref{fig:balanced_product_torus}.
\begin{figure}[h]
    \centering
    \input{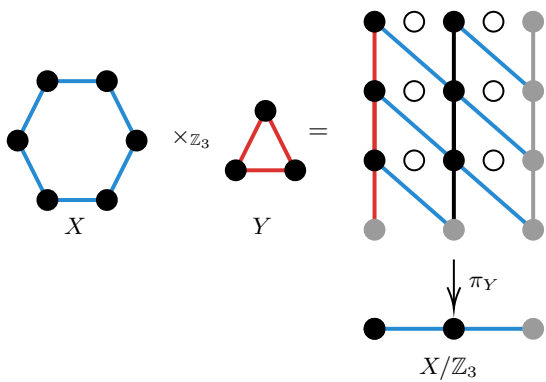}
    \caption{The balanced product $X\times_H Y$ of the length-6 and length-3 cyclic graphs~$X$ and~$Y$ over the cyclic group $H=\mathbb{Z}_3$. 
    This gives a twisted $3\times 2$ tessellation of a torus. The balanced product is a fiber bundle $\pi_Y:X\times_H Y\to X/H$ over the quotient graph~$X/H,$ a length-2 cyclic graph, with fiber~$Y.$ 
    Grey edges and vertices are only included to visualize the periodic boundaries. White circles indicate the six faces.
    When interpreted as a homological quantum code, the picture visualizes the code $C(X\times_H Y)=C(X)\otimes_H C(Y)$ which is a twisted version of the toric code with parameters $[[12,2,3]].$ Here physical qubits correspond to edges, while $X$- and $Z$-checks correspond to vertices and faces.
    The code is also a fiber bundle code $C(X/H)\otimes_\phi C(Y)$ for an appropriate connection $\phi.$
    }
    \label{fig:balanced_product_torus}
\end{figure}

\subsection{Balanced Products of Vector Spaces}
The balanced product in topology has an analogue in linear algebra.
For vector spaces $V,W$ with a linear right, respectively left action, of $H$ the \emph{balanced product} of~$V$ and~$W$ is defined as the quotient
$$V\otimes_{H}W= V\otimes W/\langle vh\otimes w-v\otimes hw\rangle.$$
To differentiate tensors in $V\otimes W$ and $V\otimes_{H}W$ we denote the former by $v\otimes w$ and the latter by $[v\otimes w]$.
Note that $[vh\otimes w]=[v\otimes hw]$.

The balanced product can be understood as the tensor product 
$$V\otimes_{H}W=V\otimes_{\mathbb{F}_2[H]} W$$
over the group algebra 
\begin{align*}
\mathbb{F}_2 H = \left\lbrace \sum_{h\in H} a_h\, h \mid a_h \in \mathbb{F}_2 \right\rbrace.
\end{align*}
If $H$ is finite and of odd order, then averaging over all group elements defines an isomorphism
\begin{align*}
V\otimes_H W\stackrel{\sim}{\to} (V\otimes W)^H,\, [v\otimes w]\mapsto \sum_{h\in H} vh\otimes h^{-1}w
\end{align*}
with $(V\otimes W)^H$ being the set of elements in $V\otimes W$ that are invariant under the anti-diagonal action
$h\cdot(v\otimes w)=vh\otimes h^{-1}w.$

If $W=\mathbb{F}_2$ is equipped with the trivial action of $H,$ then
$$V\otimes_H\mathbb{F}_2=V/\langle vh-v\rangle = V_H$$
is the space of \emph{coinvariants} of $V$ under the action of $H.$

\subsection{Balanced Products of Chain Complexes}\label{sec:balancedproductcomplex}

\subsubsection{Definition}
We can extend the definition from vector spaces to chain complexes.
Let $C,D$ be chain complexes with a linear right, respectively left, action of $H$. 
By definition, this means that the action of~$H$ commutes with the differentials.
Similarly, as for the double complexes discussed in \Cref{sec:doublecomplexes} we can form the \emph{balanced product double complex}~$C\boxtimes_HD$ via
\begin{gather*}
(C\boxtimes_{H} D)_{p,q}=C_{p}\otimes_{H} D_{q},\\
\del^{v}=\del^{C}\otimes\id_{D} \text{ and } \del^{h}=\id_{C}\otimes\del^{D}.
\end{gather*}
Here, by abusing notation we denote $\del^{v}$ and $\del^{h}$ the induced differential on the quotients $C_{p}\otimes_{H} D_{q}$ of $C_{p}\otimes D_{q}.$

The \emph{balanced product complex} is the total complex of the balanced product double complex
\begin{gather*}
C\otimes_{H} D=\Tot(C\boxtimes_{H} D).
\end{gather*}
If the balanced product complex is equipped with a natural basis it defines quantum codes which we call \emph{balanced product codes}.
See \Cref{sec:stabilizerquantumcodes} on how to obtain quantum codes from complexes.

{\em Remark}---An important special case where the balanced product complex has a natural basis is when each of the vector spaces~$C_i$ and~$D_i$ are equipped with bases~$X_i$ and~$Y_i$ and the action of~$H$ restricts to an action on these bases.
In this case, we denote by $[x,y]=\{(xh,h^{-1}y) \mid h\in H\}$ the orbit of $(x,y)\in X_i\times Y_j$ under the anti-diagonal action of $H$ and by $X_i\times_H Y_j=\{[x,y] \mid x\in X_i, y\in Y_j\}$ the set of orbits, see \Cref{sec:motivationfromtopology}.
Then a natural basis of $(C\otimes_H D)_n$ is given by the disjoint union
$$(X\times_H Y)_n=\biguplus_{i+j=n} X_i\times_H Y_j$$
and the differential can be expressed in terms of basis vectors by
$$\del([x,y])=[\del(x),y]+[x,\del(y)]=\sum_{x'\in\del(x)}[x',y]+\sum_{y'\in\del(y)}[x,y'].$$
Hence, a matrix representation of the differentials~$\del$ can be obtained by choosing an ordering of the basis vectors.

If $H$ is the trivial group, the balanced product is just the tensor product.
Hence, balanced product complexes (codes) yield a generalization of the tensor (or hypergraph) product.

Assuming that $H$ is a finite group of odd order we obtain the following version of the Künneth formula.
\begin{lemma}\label{lem:kuennethbalanced}
Let $H$ be a finite group of odd order acting on~$C$ and~$D$ from the left and right, respectively.
There is an isomorphism
\begin{gather*}
H_{n}(C\otimes_{H} D) \cong\bigoplus_{p+q=n} H_{p}(C)\otimes_{H} H_{q}(D).
\end{gather*}
\end{lemma}
\begin{proof}
To see this, note that the complex $C\otimes_H D$ is obtained from the complex $C\otimes D$ by passing to the coinvariants under the anti-diagonal action of $H$.
Since $H$ is finite and of odd order, taking coinvariants is naturally equivalent with taking invariants.
Since taking invariants commutes with taking kernels and taking coinvariants commutes with taking cokernels, they commute with passing to homology.
The statement follows from the Künneth formula for tensor products.
\end{proof}

\subsubsection{Relation to fiber bundle codes}
In special cases the balanced product complex is an instance of a fiber bundle complex introduced by Hastings--Haah--O'Donnell, see  \Cref{sec:fiberbundlecodes}.
Namely, if $C$ is a two-term complex, the action of $H$ restricts to a free action on the bases of each $C_i$ and $H$ is abelian, then there is a connection $\varphi$ such that $$C\otimes_H D= B\otimes_\varphi F$$
where $B_i=C_i/\langle ch -c\rangle$ and $F=D.$ We skip the details here and will elaborate on this in a special case later on.

\subsubsection{Relation to lifted product codes}\label{sec:defliftedproducts}
If the complex $D$ also fulfills the conditions of $C$ from the last paragraph and $H$ is abelian, then the balanced product specializes to a so called \emph{lifted product} (LP) or \emph{generalized hypergraph product} (GHP)
$$C\otimes_H D= \operatorname{LP}(\del^C,\del^D)$$
introduced by Panteleev--Kalachev in~\cite{panteleev2019degenerate,panteleev_almostlinear}.
Since we will make use of some of their results, let us elaborate on this.
Denote by $R=\mathbb{F}_2 H$ the group algebra of~$H$. 
Using the free action of~$H$ on the bases of~$C_i$ and~$D_i$ one can write the complexes~$C$ and~$D$ in the form
\begin{center}
	\begin{tikzcd}[row sep=1pt]
	C=(R^n \arrow[r,"\del^C"] & R^m)\\
	D=(R^k \arrow[r,"\del^D"] & R^l)
	\end{tikzcd}
\end{center}
for suitable $n,m,k$ and $l.$ Hence,~$\del^C\in R^{n\times m}$ and~$\del^D\in R^{k\times l}$ can be identified with matrices with entries in the algebra~$R$. 
The lifted product code is defined in terms of the matrices
$$M_1=[\del_C\otimes I_k, I_n\otimes \del^D] \text{ and } M_2=\begin{bmatrix}
 I_m\otimes \del_D  \\
 \del^C\otimes I_l \\
\end{bmatrix}$$
where $\otimes$ denotes the Kronecker product of matrices over $R.$ The matrices $M_i$ have entries in the algebra~$R$.

The action of $R$ on itself by left multiplication (\emph{regular representation}) defines an embedding $R\subset \mathbb{F}_2^{|H|\times |H|}$. 
In the case that $H=\mathbb{Z}_\ell$ is the cyclic group, the matrices in the image of this embedding are called \emph{circulant matrices}.

One can replace the entries of the matrices in $M_i$ by the corresponding matrices in $\mathbb{F}_2^{|H|\times |H|}$. 
The resulting matrices form a chain complex over $\mathbb{F}_2$ denoted by $\operatorname{LP}(\del^C,\del^D).$
It is easy to see that this complex agrees with the balanced product complex $C\otimes_H D$.

\subsection{Construction of Quantum Codes from Tanner Codes and Graph Codes}
Now, we combine the ideas of Tanner codes introduced in \Cref{sec:expander} and balanced products of chain complexes from the last section.
For this, let $X$ be an $s$-regular graph with a free action of a group $H,$ that is, $H$ acts freely on the set of vertices and edges. Further we impose the quotient condition on the action of $H$ on $X$, that is, we assume that there are no edges connecting vertices of the form $v$ and $vh$, see \Cref{sec:balancedproductsgraphs}. Let~$L$ be a local $[s,k,d]$-code and $\Lambda=\{\Lambda_v\}_{v\in X_{0}}$ a labeling, see \Cref{sec:localsystems}. 
We choose $\Lambda$ in an $H$-invariant way, that is, we assume that
\begin{equation}\label{eq:invariantlabeling}
\Lambda_{vh}(eh)=\Lambda_{v}(e) .
\end{equation}
\subsubsection{Definition in terms of the balanced product}
We will consider the balanced product double complex
$C(X,L)\boxtimes_H C(Y)$
\begin{center}
\begin{tikzcd}
C_{1}(X)\otimes_H C_{1}(Y) \arrow[d, " \id\otimes\del"'] \arrow[r, "\del\otimes \id"] &(C_{0}(X)\otimes L_0)\otimes_H C_{1}(Y)\arrow[d, "\id\otimes\id\otimes\del"]  \\
 C_{1}(X)\otimes_H C_{0}(Y) \arrow[r, "\del\otimes \id"']          & (C_{0}(X)\otimes L_0)\otimes_H C_{0}(Y)
\end{tikzcd}
\end{center}
and the associated balanced product complex
\begin{gather*}
C(X,L)\otimes_H C(Y).
\end{gather*}
Here $C(X,L)=C(X,L,\Lambda)$ denotes the Tanner code as defined in \Cref{sec:localsystems}.

Our construction of quantum codes is based on this definition, using special choices of $H,$ $X$, $L$ and $Y$ later on. 

\emph{Remark}---In order to convert the above diagram into matrices we provide bases and show how they are mapped by the differentials. Choosing an ordering of the basis elements yields the parity check matrices of the CSS code.

The vector spaces $C_i(X)$ and $C_i(Y)$ have the canonical basis $X_i$ and $Y_i$, as discussed in \cref{sec:complexes}. Further, let $\mathcal{B}_{L_0}$ denote the basis of checks of the local code~$L$.
The three bases of the total complex are: $X_{1}\times_HY_1$, $X_{1}\times_HY_0\cup X_{0}\times_HY_1\times \mathcal{B}_{L_0}$ and $X_{0}\times_HY_0\times \mathcal{B}_{L_0}.$ We denote elements in $X_{1}\times_HY_1$ by $[e,e'],$ of $X_{1}\times_HY_0$ by $[e,v']$, of $X_{0}\times_HY_1\times \mathcal{B}_{L_0}$ by $([v,e'],c)$ and of $X_{0}\times_HY_0\times \mathcal{B}_{L_0}$ by $([v,e'],c).$ The square brackets denote orbits under the anti-diagonal action of $H.$ So, for example, $[e,e']$ denotes the orbit $\{(eh,h^{-1}e') \mid h\in H\}$, see \Cref{sec:motivationfromtopology}. 
For convenience we swapped the positions of the checks of~$L$ and the cells of~$Y$.
The first boundary is then given by
 $$\tilde{\del}_1([e,e'])=\sum_{v'\in \del e'} [e,v'] + \sum_{v\in \del e}\sum_{c\in \del^L \Lambda_{v}(e)} ([v,e'], c)$$
and the second by
\begin{align*}
    \tilde{\del}_0([e,v']) &= \sum_{v\in \del e}\sum_{c\in \del^L \Lambda_{v}(e)} ([v,v'], c)\\
    \tilde{\del}_0([v,e'],c) &= \sum_{v'\in \del e'} ([v,v'], c).
\end{align*}

The example in \Cref{fig:balanced_product_torus} visualizes the balanced product code of two repetition codes.
Each repetition code is associated to a cyclic graph, interpreting the edges as bits and the vertices with checks with support on their adjacent edges.
The result is a quantum code with weight-4 $X$- and $Z$-checks which are associated to the six vertices and six faces, respectively. 
It is in fact a twisted version of the toric code with distance 3, however, the number of edges/qubits is 12 instead of 18.

The construction above is a special case of the balanced product of two complexes, see \Cref{sec:balancedproductcomplex}. 
There are many interesting variants of these definitions, as for example the balanced product of two Tanner codes, which we will not discuss here.

\subsubsection{Definition in terms of  fiber bundles}\label{sec:fiberbundledef}
As in the previous section, there is also a coordinatized definition of $C(X,L)\boxtimes_H C(Y)$ and $C(X,L)\otimes_H C(Y).$ For this, note that $X/H$ is still an $s$-regular graph and by \Cref{eq:invariantlabeling} the labeling $\Lambda$ descends to $X/H.$ We can hence also consider the Tanner code
\begin{center}
	\begin{tikzcd}
	C(X/H,L)=(C_1(X/H) \arrow[r,"\del"] & C_0(X/H)\otimes L_0)
	\end{tikzcd}
\end{center}
and \emph{fiber bundle double complex} $C(X/H,L)\boxtimes_{\phi_R}C(Y)$  
\begin{center}
\begin{tikzcd}
C_{1}(X/H)\otimes C_{1}(Y) \arrow[d, " \id\otimes\del"'] \arrow[r, "\del_{\phi_R}"] &C_{0}(X/H)\otimes L_0\otimes C_{1}(Y)\arrow[d, "\id\otimes\id\otimes\del"]  \\
 C_{1}(X/H)\otimes C_{0}(Y) \arrow[r, "\del_{\phi_R}"']          & C_{0}(X/H)\otimes L_0\otimes C_{0}(Y)
\end{tikzcd}
\end{center}
where
\begin{multline*}
\del_{\phi_R}(\mathcal{E}\otimes y)=\\ \mathcal{V}\otimes \del^L \Lambda_{\mathcal{V}}(\mathcal{E})  \otimes y +\mathcal{V}'\otimes \del^L \Lambda_{\mathcal{V}'}(\mathcal{E})\otimes \phi_R(\mathcal{E})y
\end{multline*}
for an oriented edge $\mathcal{E}=(\mathcal{V},\mathcal{V}').$ The associated \emph{fiber bundle complex}  is obtained as the total complex $$C(X/H,L)\otimes_{\phi_R}C(Y)=\Tot(C(X/H,L)\boxtimes_{\phi_R}C(Y)).$$ 
Again, it is straightforward to see that one can equate
\begin{align*}
C(X,L)\boxtimes_{H}C(Y)&=C(X/H,L)\boxtimes_{\phi_R}C(Y)\text{ and}\\
C(X,L)\otimes_{H}C(Y)&=C(X/H,L)\otimes_{\phi_R}C(Y).
\end{align*}
Moreover, one can also interpret $C(X,L,\Lambda)\otimes_{H}C(Y)$ as a lifted product code, see Section \ref{sec:defliftedproducts}.
It is helpful to have all perspectives in mind.
\subsubsection{(Co-)Homology}
The (co-)homology groups of the balanced product complex $C(X,L)\otimes_{H}C(Y)$ can be calculated using the Künneth formula for balanced products, see \Cref{lem:kuennethbalanced}. 

We will make this explicit here. We denote elements in $H_1(C(X,L)\otimes_{H}C(Y))$ by $[(u,v)]$ where
\begin{align*}
u\in C_1(X,L)\otimes_H C_0(Y)\text{ and } v\in C_0(X,L)\otimes_H C_1(Y)
\end{align*}
and refer to $u$ as the \emph{horizontal} and $v$ as the \emph{vertical} part. The Künneth formula implies that $H_1(C(X,L)\otimes_{H}C(Y))$ is generated by the complementary subspaces of \emph{horizontal} homology classes $$H^h_1(C(X,L)\otimes_{H}C(Y)),$$ which admit a representative of the form $(u,0)$
and \emph{vertical} homology classes $$H^v_1(C(X,L)\otimes_{H}C(Y)),$$ which admit a representative of the form $(0,v).$ 
We denote the corresponding projection maps by 
\begin{align*}
p^h&: H_1(C(X,L)\otimes_{H}C(Y))\to H^h_1(C(X,L)\otimes_{H}C(Y))\\
p^v&: H_1(C(X,L)\otimes_{H}C(Y))\to H^v_1(C(X,L)\otimes_{H}C(Y)).
\end{align*}
The maps can be computed as
\begin{align*}
p^h([(u,v)])=[u]\text{ and }p^v([(u,v)])= [v]
\end{align*}
if $[u]$ and $[v]$ are homology classes.
To evaluate the projection maps given a general representative~$(u,v),$ note that the Künneth formula ensures that one can always find a homologous representative $(u',v')$ where~$[u']$ and~$[v']$ \emph{are} homology classes.

Note that if $H_0(C(X,L))=0,$ or equivalently, the checks of the Tanner code $C(X,L)$ are linearly independent, we have
$$H_1(C(X,L)\otimes_{H}C(Y))=H_1^h(C(X,L)\otimes_{H}C(Y))$$
and there are no non-trivial vertical homology classes.

The obvious dual definitions and statements apply to cohomology. For example, the space of \emph{horizontal} \emph{co}homology classes denoted by $$H^1_h(C(X,L)\otimes_{H}C(Y))\subseteq H^1(C(X,L)\otimes_{H}C(Y))$$
consists of cohomology classes with have a representative of the form $(u,0)$ for $$u\in C^1(X,L)\otimes_H C^0(Y)=C_1(X,L)\otimes_H C_0(Y).$$
The corresponding projections are denoted by~$p_h$ and~$p_v$.

\subsection{Balanced Product with a Circle}
We are interested in a special situation. Namely, let $H=\mathbb{Z}_\ell$ be the cyclic group and $Y=C_{\ell}$ the cycle graph of size $\ell,$ with~$\ell$ odd. 
Then the natural action of $\mathbb{Z}_\ell$ on $C_{\ell}$ by translation induces a trivial action on the homology groups $H_0(C_{\ell})=H_1(C_{\ell})=\mathbb{F}_2.$
Fix some vertex $y_0$ and edge $y_1$ in $C_{\ell}.$ Then the maps
\begin{align*}
\mathbb{F}_2[\mathbb{Z}_\ell] \to C_i(C_{\ell}),\, \sum_h a_hh\mapsto \sum_h a_hy_ih
\end{align*}
are isomorphisms for $i=1,2$.
Similarly, for every vector space~$V$ with a linear right action the maps
\begin{align*}
	V\to V\otimes_{\mathbb{Z}_\ell} C_i(C_\ell), v\mapsto [v\otimes y_i].
\end{align*}
are isomorphism for $i=1,2$.

Consider the linear maps
\begin{align*}
\iota: H_1(C(X/H,L))&\to H^h_1(C(X,L)\otimes_{\mathbb{Z}_\ell}C(C_\ell)), \\
 \left[\sum_{\mathcal{E}} a_\mathcal{E}\mathcal{E}\right] &\mapsto \left[ \sum_{\mathcal{E}} a_\mathcal{E}\sum_{e\in \mathcal{E}}[e\otimes y_0],0\right]
\end{align*}
taking a code word of the Tanner code $H_1(C(X/H,L))$ to a code word in the balanced product code $H^h_1(C(X,L)\otimes_{\mathbb{Z}_\ell}C(C_\ell))$.
Conversely we define the linear map
\begin{align*}
\pi: H^h_1(C(X,L)\otimes_{\mathbb{Z}_\ell}C(C_\ell))&\to H_1(C(X/H,L))\\
	\left[\left[\sum_e a_{e} e\otimes y_0\right],0\right]&\mapsto \left[\sum_e a_{e}eH\right]
\end{align*}
which projects a code word of the balanced product code $H^h_1(C(X,L)\otimes_{\mathbb{Z}_\ell}C(C_\ell))$ onto a code word  of the Tanner code $H_1(C(X/H,L))$.

The composition $\pi\circ\iota$ is the multiplication with $|\mathbb{Z}_\ell| = 1 \mod 2$ and hence the identity. Note that by the Künneth formula there is an isomorphism 
\begin{align*}H^h_1(C(X,L)\otimes_{\mathbb{Z}_\ell}C(C_\ell))&\cong H_1(C(X,L))\otimes_{{\mathbb{Z}_\ell}}H_0(C_\ell)\\
&=H_1(C(X,L))/\langle vh-v\rangle
\end{align*}
and the latter can naturally identified with $H_1(C(X/H,L)).$ Hence $\pi$ and $\iota$ are indeed isomorphisms by a dimension argument.

For $Y=C_\ell$ the projection map $p^h$ takes a particularly nice form when composed with~$\pi$.
Independently of finding a nice representative~$(u,v)$ where~$[u]$ is a homology class, one has
$$\pi\circ p^h\left[\left[\sum_ea_e e\otimes y_0\right],v\right]=\left[\sum_ea_eeH\right].$$
This follows easily by checking that the definition is invariant under adding boundaries.
Dually, the transposed of the maps~$\pi$ and~$\iota$ define an isomorphism
$$ H_h^1(C(X,L)\otimes_{\mathbb{Z}_\ell}C(C_\ell))\cong H^1(C(X/H,L)).$$

We subsume the result in the following theorem.
\begin{theorem}\label{thm:encodingratecircle}
The number of logical bits that can be encoded in the horizontal part of the balanced product code
\begin{align}\label{eqn:encodingratecircle}
 \dim H^h_1(C(X,L)\otimes_{\mathbb{Z}_\ell}C(C_\ell))&=\dim H_1(C(X/H,L)) 
\end{align}
 agrees with the number of logical bits in the Tanner code $C(X/H,L)$.
\end{theorem}
We note that these results can also be obtained by using the definition by fiber bundles outlined in \Cref{sec:fiberbundledef} and the discussion in \Cref{sec:fiberbundlecodes}.

\subsection{Balanced Product Subsystem Codes}\label{sec:subsystemofourthing}
For later applications we are only interested in the horizontal part of the homology.
This is because we will obtain superior distance bounds for these elements.
To achieve this we make use of the formalism of subsystem codes which we introduced in \Cref{sec:classicalandquantumcodes}.
We refer to this as the {\em horizontal subsystem balanced product code}.

We define the logical operators of $Z$-type to correspond to the non-trivial elements of $$H_1^\mathcal{L} = H^h_1(C(X,L)\otimes_{\mathbb{Z}_\ell}C(C_\ell))$$ and the gauge qubits of $Z$-type, which we disregard, to be the non-trivial elements of $$H_1^\mathcal{G} = H^v_1(C(X,L)\otimes_{\mathbb{Z}_\ell}C(C_\ell)).$$
This induces the splitting in the cohomology $$H^1_\mathcal{L} = H_h^1(C(X,L)\otimes_{\mathbb{Z}_\ell}C(C_\ell))$$ and $$H^1_\mathcal{G} = H_v^1(C(X,L)\otimes_{\mathbb{Z}_\ell}C(C_\ell)).$$

Note, that if the checks in the Tanner code~$C(X/H,L)$ are independent then the vertical (co)homology classes vanish.
In this case this procedure is unnecessary as $H_1^\mathcal{G}=0$.

\subsection{Distance Theorems}
Using the ingenious distance bounds for lifted product codes by Panteleev--Kalachev in~\cite{panteleev_almostlinear}, we are able to bound the (co-)homological distance of the balanced products between Tanner codes and cycles $C(X,L)\otimes_{\mathbb{Z}_\ell}C(C_\ell)$ in terms of expansion properties of the Tanner code~$C(X,L).$
\begin{theorem}[Panteleev--Kalachev~\cite{panteleev_almostlinear}]\label{thm:distho}
Assume that the boundary map of the Tanner code $C(X,L)$
\begin{center}
\begin{tikzcd}
C_1(X) \arrow[r,"\del"] & C_0(X)\otimes L_0
\end{tikzcd}
\end{center}
is $(\alpha_{\operatorname{ho}}, \beta_{\operatorname{ho}})$-expanding.
Let $x=(u,v)$ be a representative of a non-trivial homology class $[x]\in H_1(C(X,L)\otimes_{\mathbb{Z}_\ell}C(C_\ell)).$
\begin{enumerate}
\item If the horizontal part $$0\neq p^h(x)\in H^h_1(C(X,L)\otimes_{\mathbb{Z}_\ell}C(C_\ell))$$ of $x$ is non-trivial, then 
$$|x|=|u|+|v| \geq |X_{1}|\min\left\{\alpha_{\operatorname{ho}}/2, \alpha_{\operatorname{ho}}\beta_{\operatorname{ho}}/4\right\}.$$
\item Else, if $p^h(x)=0,$ then 
$$|x|=|u|+|v| \geq  \ell\min\left\{\alpha_{\operatorname{ho}}/(4s),\alpha_{\operatorname{ho}}\beta_{\operatorname{ho}}/(4s)\right\}.$$
\end{enumerate}
\end{theorem}
\begin{proof} The code $C(X,L)\otimes_{\mathbb{Z}_\ell}C(C_\ell)$ agrees with the a lifted product code $\operatorname{LP}(A,1+x)$ in the notation of Panteleev--Kalachev~\cite{panteleev_almostlinear} where $A=\del$ is the matrix representing the differential $\del$ of the Tanner code $C(X,L).$

The two cases in the theorem correspond to Case 1 and 2 in the proof of~\cite[Proposition 2]{panteleev_almostlinear}.
The homology class~$\pi(p^h(x))$ corresponds to~$u(1),$ and $|X_{1}|$ to $n\ell$ in the notation of~\cite{panteleev_almostlinear}.
\end{proof}
There is the following dual distance bound for the cohomology.
\begin{theorem}[Panteleev--Kalachev~\cite{panteleev_almostlinear}]\label{thm:distco}
Assume that the coboundary map of the Tanner code $C(X,L)$
\begin{center}
\begin{tikzcd}
C_1(X)  & C_0(X)\otimes L_0\arrow[l,"\delta"']
\end{tikzcd}
\end{center}
is $(\alpha_{\operatorname{co}}, \beta_{\operatorname{co}})$-expanding.
Let $x=(u,v)$ be a representative of a non-trivial cohomology class $[x]\in H^1(C(X,L)\otimes_{\mathbb{Z}_\ell}C(C_\ell)).$
\begin{enumerate}
\item If the vertical part $$0\neq p_v(x)\in H_v^1(C(X,L)\otimes_{\mathbb{Z}_\ell}C(C_\ell))$$ of $x$ is non-trivial, then 
$$|x|=|u|+|v| \geq |X_{0}|s\min\left\{\alpha_{\operatorname{co}}/2, \alpha_{\operatorname{co}}\beta_{\operatorname{co}}/4\right\}.$$
\item Else, if $p_v(x)=0,$ then 
$$|x|=|u|+|v| \geq  \ell\min\left\{\alpha_{\operatorname{co}}/(4s),\alpha_{\operatorname{co}}\beta_{\operatorname{co}}/(4s)\right\}.$$
\end{enumerate}
\end{theorem}
To subsume, we obtain the following corollary.
\begin{corollary}\label{cor:distanceboundssybsystemcode}
The stabilizer/subsystem code of $C(X,L)\otimes_{\mathbb{Z}_\ell}C(C_\ell)$ defined by $H_1^\mathcal{L} = H^h_1(C(X,L)\otimes_{\mathbb{Z}_\ell}C(C_\ell))$, see \Cref{sec:subsystemofourthing}, is an $[[N,K,D_X,D_Z]]$ code where
\begin{align*}
N &= 3\,|X_{1}|\\
K &= \dim H_1(C(X/\mathbb{Z}_\ell,L))\geq (2k_L/s-1)|X_{1}|/\ell\\
D_Z &\geq |X_{1}|\min\left\{\alpha_{\operatorname{ho}}/2, \alpha_{\operatorname{ho}}\beta_{\operatorname{ho}}/4\right\}\\
D_X &\geq \min\left\{\alpha_{\operatorname{co}}|X_{1}|, \alpha_{\operatorname{co}}|X_{1}|/2,\ell\alpha_{\operatorname{co}}/(4s),\ell\alpha_{\operatorname{co}}\beta_{\operatorname{co}}/(4s)\right\}
\end{align*}
in the notation of \Cref{thm:distho} and \Cref{thm:distco}.
\end{corollary}
\begin{proof}
The number of physical qubits~$N$ is given by 
$$\dim\, (C(X,L)\otimes_{\mathbb{Z}_\ell}C(C_\ell))_1 = |X_{1}|+s |X_{0}|= 3|X_{1}|.$$
The number of logical qubits encoded in the horizontal subsystem is the number of bits encoded in the Tanner code~$C(X/\mathbb{Z}_\ell, L),$ see \Cref{thm:encodingratecircle}, which is bounded below by 
$$k_{C(X/\mathbb{Z}_\ell, L)}\geq (2k_L/s-1)\, |(X/\mathbb{Z}_\ell)_1|$$
via \Cref{eqn:sipser_spielman_k}.
Now $\mathbb{Z}_\ell$ acts freely on $X$ by assumption and hence $|(X/\mathbb{Z}_\ell)_1|=|X_{1}|/\ell.$

The bound for the homological distance $D_Z$ follows from Case 1 of \Cref{thm:distho}, since the homological distance of the subsystem code is the minimum distance of representatives of homology classes which do not vanish when projected onto $H_1^\mathcal{L}=H_1^h(C(X,L)\otimes_{\mathbb{Z}_\ell}C(C_\ell)).$

The bound for~$D_X$ follows from combining Case 1 and 2 in \Cref{thm:distco}.
\end{proof}

\subsection{Concrete Example of a Balanced Product Code}
Before we discuss the explicit construction of a code family using the balanced product which breaks the $\operatorname{polylog}(N)\sqrt{N}$ distance bound in the next section, we will briefly give a concrete example of a balanced product code.

First, let us consider the hyperbolic $\{3,7\}$ tessellation (cf.~\Cref{fig:klein_quartic}) of a genus-14 surface, which is defined by the translation $\tau = r^2 s^{-1} r^2 s^{-1} r^2 s^{-1} r^2 s^{-1} r^2 s^{-1} r^2 s^{-1}$ (cf.~Appendix~\ref{sec:hyperbolic_compact}).
This tessellation gives a degree-7 graph~$X$ with 156 vertices and 546 edges.
For the local code~$L$ we choose the $[7,4,3]$ Hamming code with parity check matrix
\begin{align*}
    \begin{bmatrix}
        1 & 0 & 1 & 0 & 1 & 0 & 1\\
        0 & 1 & 1 & 0 & 0 & 1 & 1\\
        0 & 0 & 0 & 1 & 1 & 1 & 1
    \end{bmatrix}.
\end{align*}
This defines an expander code $C(X,L)$ which has $546$ bits, identified with the edges of $X$, and each of the $156 \times 3 =
468$ checks has weight~4.
As $\dim H_0(C(X,L)) = 0$ all checks are linearly dependent, so that $C(X,L)$ encodes $\dim H_1(C(X,L)) = (2k_L/s-1)\, |(X/\mathbb{Z}_\ell)^1| = (2\times 4/7-1)\, |X_{1}|/13 = 78$ bits.
We note that $\lambda_2(X) > 5 > d_L$ so that we do not have a non-trivial lower bound on the distance of~$C(X,L)$.
However, this is a common issue for expander codes with low check weight and such codes generally still perform well in practice~\cite{dowling2017fast}.

The orientation-preserving automorphism group $G^+/N_\tau$ of the surface has an (up to conjugation) unique cyclic subgroup of order $\ell = 13$.
The balanced product code $C(X,L)\otimes_{\mathbb{Z}_{13}}C(C_{13})$ has $N = 1014$ qubits and encodes $K = \dim H_1(C(X/\mathbb{Z}_{13},L)) = 78/13 = 6$ qubits.
The $X$-stabilizers have weight $4+2=6$ coming from the combined check weight of the Hamming code~$L$ and the repetition code~$C(C_{13})$.
The weights of the $Z$-stabilizers are $4, 5, 6, 7, 8$, resulting from the combined weight of the respective check- and bit-degrees.
A Monte Carlo brute-force search showed that the $X$-distance~$D_X$ is at most~13.
This is expected as $\ell=13$ is exactly the length of the cycle in the product.
The lowest-weight $Z$-logical that we could find using the same method had weight~18.



\section{Explicit Examples}\label{sec:explicit}

The goal of this section is to prove the following theorem.
\begin{theorem}\label{thm:explicitfamily}
There exists an explicit construction of a family of $[[N,K,D_X,D_Z]]$ LDPC quantum codes which encode $K \in \Theta(N^\frac{2}{3})$ logical qubits, with $X$-distance $D_X \in \Omega(N^\frac{1}{3})$ and $Z$-distance $D_Z \in \Theta(N)$.
\end{theorem}
The distance balancing of~\cite{hastings2016weight} and~\cite{ramanujancomplexesone} respects the splitting of the homology group described in \Cref{sec:subsystemofourthing}.
Applying the distance balancing procedure of~\cite{ramanujancomplexesone} using classical $[N^\frac{2}{3},\Theta(N^\frac{2}{3}),\Theta(N^\frac{2}{3})]$-codes we obtain the following corollary.
\begin{corollary}\label{cor:explicitfamily}
There exists an explicit construction of a family of $[[N,K,D]]$ LDPC quantum codes which encode $K \in \Theta(N^\frac{4}{5})$ logical qubits with distance $D \in \Omega(N^\frac{3}{5})$.
\end{corollary}

\subsubsection{Expander graphs $X$}
The graphs~$X$ will be LPS-expanders defined in \Cref{sec:explicit_expanders}.
Choose primes~$p$ and~$q$ such that $$\left(\dfrac{p}{q}\right) = -1$$ so that the graphs $X=X_{p,q}$ are Cayley graphs of the group~$\operatorname{PGL}(2,q)$.
Hence we have $$|X_0| = |\operatorname{PGL}(2,q)| = q(q^2-1),$$ their degree is $s=p+1$ and their second largest eigenvalues satisfy the bound $\lambda_2 \leq 2\sqrt{s-1}$.

\subsubsection{Subgroups $H$}
The subgroups~$H$ for the balanced product need to be cyclic subgroups which act freely on the graphs~$X.$ 
Furthermore, to obtain regular quotient graphs, we need to assume the quotient condition, that is,  there is no edge connecting vertices~$g$ and~$gh$ for~$h\in H$ (see \Cref{sec:balancedproductsgraphs}). 
Since all edges in~$X$ are of the form~$(g,sg)$ for~$s\in S_{p,q}$ this is equivalent to 
\begin{align}\label{eqn:quotientcondition}
	S_{p,q}\cap gHg^{-1}=\emptyset \text{ for all } g\in \operatorname{PGL}(2,q).
\end{align}
The quotient condition implies that we act freely on the edges~$X_1$.

We claim that these properties are fulfilled by the subgroup of unipotent upper-triangular matrices 
\begin{align*}
	H = \left\lbrace \begin{bmatrix}
	1 & x \\
	0 & 1
	\end{bmatrix} \mid x \in \mathbb{F}_q \right\rbrace \subset \operatorname{PGL}(2,q).
\end{align*}
First~$H$ is isomorphic to the additive group of $\mathbb{F}_q$ and hence a cyclic group of order $q$. Since $X$ is a Cayley graph $H$ acts freely on the vertices~$X_0$.

To show that no element in $H$ is conjugate to an element in~$S_{p,q}$, see \Cref{eqn:quotientcondition}, we argue via determinants. 
Note that the determinant of an element $g\in \operatorname{PGL}(2,q)$ is defined up to multiplication with a square, so $\det(g)\in \mathbb{F}_q^\times/(\mathbb{F}_q^\times)^2.$ We have
$$\det(ghg^{-1})=1\in \mathbb{F}_q^\times/(\mathbb{F}_q^\times)^2\text{ and }\det(s)=p\in \mathbb{F}_q^\times/(\mathbb{F}_q^\times)^2$$
for all $h\in H,$ $g\in \operatorname{PGL}(2,q)$ and $s\in S_{p,q}.$ 
However, by assumption $p$ is not a square modulo $q$ and hence $1\neq p$ in $\mathbb{F}_q^\times/(\mathbb{F}_q^\times)^2.$

\subsubsection{Local codes $L$}
In order for the Tanner code $C(X,L)$ to be a code with a linear number of encoded bits and linear distance we need it to satisfy the constraints given by \Cref{eqn:sipser_spielman_k} and \Cref{eqn:expandercode_local_distance_demand}, namely that it has to have rate $k_L/2 > 1/2$ and distance larger than the second eigenvalue of the graph $d_L > \lambda_2$.
Furthermore, in order to be able to apply \Cref{thm:expanderbitdegree} we need that the dual code~$L^\perp$ has distance larger than the second eigenvalue of the graph $d_L > \lambda_2$ as well.

\Cref{thm:gilbertvarshamovplus} guarantees the existence of such a local code~$L$ with the lower bound $d_L,d_{L^\perp} \geq \delta s$ for some $\delta \in (0,0.11)$.
Although \Cref{thm:gilbertvarshamovplus} does not give a construction of this code~$L$, we can find it by brute force in~$O(1)$ time with respect to the size of~$X$.

\begin{proof}[Proof of \Cref{thm:explicitfamily}]
With the discussion above, choose $p=401$, $\delta = 0.1$, $\alpha_{\operatorname{ho}} = 10^{-3}$, $\alpha_{\operatorname{co}} = 10^{-5}$.
Then using \Cref{thm:expanderviolatedchecks} and \Cref{thm:expanderbitdegree} we get that $\beta_{\operatorname{ho}}$ and $\beta_{\operatorname{co}}$ are bounded below by positive constants.
The family of codes are the horizontal subsystem balanced product codes.
Using \Cref{cor:distanceboundssybsystemcode} and letting $q\rightarrow \infty$ we obtain a code family with the desired properties.
We note that in our construction we essentially took Kronecker products of sparse matrices so that the resulting parity check matrices are sparse as well which in turn implies that the family is LDPC.
\end{proof}

\section{Conclusion}\label{sec:conclusion}
We have demonstrated that balanced products of classical codes give rise to explicit families of LDPC quantum codes with exceptional asymptotic properties.
In this manuscript we focused on balanced products with cycle graphs. 
This prevents the number of logical qubits~$K$ and cohomological distance~$D_X$ from being linear. 
We conjecture that \emph{good} LDPC quantum codes with linear number of logical qubits and distance can obtained by replacing cycle graphs by good classical codes.

\begin{conjecture*}
The balanced product of two good classical LDPC codes over groups of order $\Theta(N)$ gives rise to a family of \emph{good} LDPC quantum codes.

In particular, the horizontal subsystem codes associated to the balanced products $C(X_{p,q},L)\otimes_{\operatorname{PGL}(2,q)}C(X_{p,q},L)^*$ form a family of \emph{good} LDPC quantum codes.
\end{conjecture*} 

We were able to show that these codes have constant encoding rate.
We note that trivial counting arguments do not apply due to a potentially non-free action of the symmetry group on the (co)homology classes, but it requires a more careful analysis of the double complex.
We leave the question if they also attain linear distance as an open problem.

A further line of research is finding efficient decoding algorithms for our codes. There are a number of interesting and efficient decoding algorithms for expander codes. It seems reasonable that they can be adapted to our setting as done in~\cite{fawzi2018constant}.


%

\appendices
\section{Spectral Sequences}\label{sec:spectralsequence}
In general, the computation of the homology of a total complex $\Tot(E)$ can be quite subtle and subject of the powerful mathematical formalism of \emph{spectral sequences}, which we want to very briefly explore. For a nice reference, see \cite[Section III.7.3]{gelfand_methods_2003}.
Simply stated, the formalism is an iterative procedure to obtain better and better approximations of the homology of $\Tot(E).$
Precisely, the spectral sequence of the double complex $E$ consists of so called \emph{pages} of arrays of vector spaces $E^{r}_{p,q}$ labeled by a non-negative integer $r,$ such that \begin{gather*}
\bigoplus_{p+q=n}E^{r}_{p,q} \text{ ``converges to'' } H_{n}(\Tot(E))
\end{gather*}
 as $r\rightarrow\infty$ and one suggestively writes
 $$E^{r}_{p,q}\Rightarrow H_{n}(\Tot(E)), \text{ for } p+q=n.$$
Furthermore, there is an iterative procedure to go from one page to the next. 
Namely, the $r$-th page of a spectral sequence is equipped with differentials $$\del_{r,p,q}:E^{r}_{p,q}\to E^{r}_{p-r,q+r-1}$$ such that the next page $E^{r+1}_{p,q}$ is obtained by taking homology along those differentials. 
 
Spectral sequences can be very intricate. 
However, in many circumstances one can show that all differentials on a certain page~$r_{0},$ where~$r_{0}$ is hopefully very small, vanish.  
In this case, one says that the spectral sequence \emph{degenerates on page~$r_{0}$}. 
Then, all following pages will be the same and there are isomorphisms
$$\bigoplus_{p+q=n}E^{r}_{p,q}\cong H_{n}(\Tot(E))$$
for all~$r\geq r_{0}.$

Let us briefly discuss the first three pages of the spectral sequence of the double complex $E.$
The $0$-th page of is simply given by $E$ itself $$E^{0}_{p,q}=E_{p,q}.$$ 
Its differential $$\del_{0}=\del^{v}: E^{0}_{p,q}=E_{p,q}\rightarrow E_{p,q-1}=E_{p,q-1}^{0}.$$
is the vertical differential of $E.$
\begin{figure}
\begin{center}
\begin{tikzcd}
 {E_{p,q}} \arrow[d, "{\del^v=\del^{0}}"'] & {E_{p-1,q}} \arrow[d, "{\del^v=\del^{0}}"] \\
 {E_{p,q-1}} & {E_{p-1,q-1}}   &       
\end{tikzcd}
\end{center}
\caption{A part of the $0$-th page of the spectral sequence $E^{r}_{p,q}\Rightarrow H^{n}(\Tot(E))$.}
\end{figure}
The first page is given by the homology of this differential
$$E^{1}_{p,q}=H_{q}(E_{p,\bullet},\del^{v}).$$
Since  $\del^{v}$ and $\del^{h}$ commute, the horizontal differential $\del^{h}$ induces maps
\begin{align*}\del_{1}=\del^{h}: E^{1}_{p,q}=&H_{q}(E_{p,\bullet},\del^{v})\\&\to E^{1}_{p-1,q}=H_{q}(E_{p-1,\bullet},\del^{v})\end{align*}
which constitute the differential on the first page.
\begin{figure}
\begin{center}
\begin{tikzcd}
H_{q}(E_{p,\bullet},\del^{v}) \arrow[r, "{\del^h}"] & H_{q}(E_{p-1,\bullet},\del^{v})\\
H_{q-1}(E_{p,\bullet},\del^{v})\arrow[r, "{\del_{p,q-1}^h}"']  &H_{q-1}(E_{p-1,\bullet},\del^{v})  &       
\end{tikzcd}
\end{center}
\caption{A part of the first page of the spectral sequence $E^{r}_{p,q}\Rightarrow H^{n}(\Tot(E))$.}
\end{figure}
Iterating once more, the second page is obtained by taking homology along this differential
$$E^{2}_{p,q}=H_{p}(H_{q}(E_{\bullet,\bullet},\del^{v}),\del^{h}).$$
In other words, the second page is obtained by first taking vertical and then horizontal homology of the double complex. It is already quite involved to write down the differential on the second page.
\section{Hyperbolic Surfaces and their Symmetries}\label{sec:hyperbolic}

\subsubsection{Hyperbolic geometry}
The geometry of space with constant negative curvature is called hyperbolic.
This geometry can be realized in different models consisting of a set of points and a hyperbolic distance function.
We will consider the \emph{Poincar\'e disc model} in which the set of points is given by the interior of a disc of unit radius
\begin{align*}
	\mathbb{H}^2 = \lbrace z\in \mathbb{C} \mid \lVert z \rVert < 1 \rbrace .
\end{align*}
and the distance between two points~$x,y \in \mathbb{H}^2_{d}$ is given by
\begin{align*}
	\dist(x,y) = \cosh^{-1} \left( 1+\frac{2\, \lVert x - y \rVert^2}{(1-\lVert x \rVert^2)(1-\lVert y \rVert^2)} \right) .
\end{align*}
Lengths are highly distorted but the angles in which lines intersect are faithfully represented.
Shortest lines (geodesics) connecting two points are given by circular arcs which intersect the unit circle in right angles.

A fact of hyperbolic trigonometry which we will use later is the following: 
Consider a triangle in~$\mathbb{H}^2$ with internal angles~$\alpha, \beta, \gamma$.
It must hold that
\begin{align}\label{eqn:hyperbolic_triangle_angles}
	\alpha + \beta + \gamma < \pi .
\end{align}
A proof of this fact and more background on hyperbolic geometry can be found in~\cite{ratcliffe2006foundations}.

\subsubsection{Isometry group}\label{sec:hyperbolic_isometries}
Bijective maps $\iota : \mathbb{H}^2 \rightarrow \mathbb{H}^2$ which leave the distance between any two points invariant, i.e.
\begin{align}\label{eqn:isometries}
\dist(\iota(x),\iota(y)) = \dist(x,y)
\end{align}
for all $x,y\in \mathbb{H}^2$, are called \emph{hyperbolic isometries}.
Just as for the euclidean plane~$\mathbb{E}^2$, they can be reflections, rotations, translations or combinations thereof.
The set of all isometries forms a group under composition called the isometry group~$\operatorname{Isom}(\mathbb{H}^2)$.

Note that any translation (in hyperbolic or euclidean space) can be written in terms of two reflections along suitably chosen parallel lines.
The same holds for rotations where the two lines intersect at the fixed-point of the rotation.
In fact~$\operatorname{Isom}(\mathbb{E}^2)$ and~$\operatorname{Isom}(\mathbb{H}^2)$ are both generated by reflections.
We call an isometry \emph{orientation-preserving} if it consists of an even number of reflections.
Orientation-preserving isometries form a subgroup of index~$2$ which we denote $\operatorname{Isom}^+(\mathbb{H}^2)$.

In the Poincar\'e model isometries can be realized in terms of a matrix group.\footnote{It is more common to use the so-called upper-half plane model to define these isometries. However, the half-plane model is equivalent to the Poincar\'e disc model via the Cayley transformation.}
We define the action of matrices on $\mathbb{H}^2$ via fractional linear transformations
\begin{align}\label{eqn:fractional_linear_trafo}
\begin{bmatrix}
a & b \\
c & d
\end{bmatrix}
\cdot z = \frac{az+b}{cz+d} .
\end{align}
These transformations are also called \emph{M\"obius transformations}.
By enforcing \Cref{eqn:isometries} we see that real $2\times 2$-matrices with determinant~$+1$ leave the distance function~$\dist_h$ invariant.
They therefore describe isometries of~$\mathbb{H}^2$.
Furthermore, it is easy to see that multiplying a matrix with~$-1$ defines the same isometry.

Let $\operatorname{SL}(2,\mathbb{R})$ be the group of real $2\times 2$-matrices with determinant~$+1$, called the \emph{special linear group}.
As the overall sign does not affect the transformation in \Cref{eqn:fractional_linear_trafo} we consider the \emph{projective special linear group} $\operatorname{PSL}(2,\mathbb{R}) = \operatorname{SL}(2,\mathbb{R}) / \{\pm 1\}$.
A classic result of hyperbolic geometry is that the fractional linear transformations of~$\operatorname{PSL}(2,\mathbb{R})$ are in fact exactly the orientation preserving isometries, i.e. we have
\begin{align*}
	\operatorname{Isom}^+(\mathbb{H}^2) \simeq \operatorname{PSL}(2,\mathbb{R}) .
\end{align*}

Finally, we note that a subgroup~$H$ of~$\operatorname{Isom}^+(\mathbb{H}^2)$ operates fixed-point free if and only if~$H$ is torsion-free, i.e. if it contains no elements of finite order.

\subsubsection{Hyperbolic tessellations}
Just as we can tessellate the euclidean plane~$\mathbb{E}^2$ with polygonal shapes, we can tessellate the hyperbolic plane~$\mathbb{H}^2$ with polygons as well.
A particular class of tessellations are \emph{regular tessellations} which are edge-to-edge coverings by copies of a regular polygon with the same number of polygons meeting at a vertex.
We will call the polygons of the tessellation \emph{faces}.
Regular tessellations are completely determined by the number of edges of each face~$r$ (r-gons) and the number of faces incident to each vertex~$s$.
This information is summarized in the \emph{Schl\"afli symbol} $\{r,s\}$.
Not every pair of numbers $r$ and $s$ is a valid tessellation due to geometric constraints.
This can be seen as follows:
subdivide a face into $2r$ triangles by cutting along all lines of symmetry.
Each triangle has internal angles~$\pi/2$,~$\pi/r$ and~$\pi/s$.
In~$\mathbb{E}^2$ the sum of these angles has to be $\pi$, so that the only possible regular tessellations of~$\mathbb{E}^2$ are tessellations by squares~$\{4,4\}$, hexagons~$\{6,3\}$ and triangles~$\{3,6\}$.
In the hyperbolic plane the restriction due to \Cref{eqn:hyperbolic_triangle_angles} allows for an infinite number of regular tessellations, namely all $\{r,s\}$ such that
\begin{align}\label{eqn:hyperbolic_schlaefli_condition}
	\frac{1}{r} + \frac{1}{s} < \frac{1}{2} .
\end{align}

\begin{figure}
	\centering
	\includegraphics[width=0.7\linewidth]{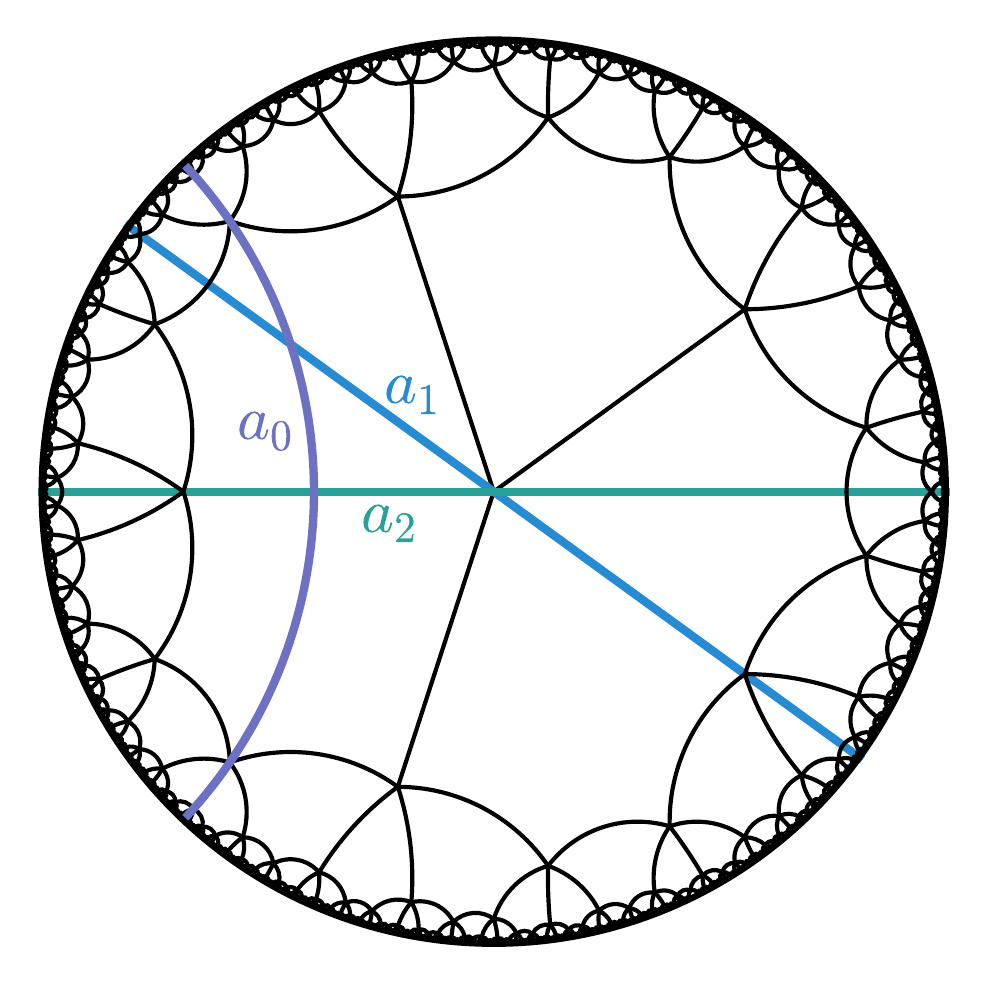} 
	\caption{Picture of a regular $\{5,5\}$ tessellation in in the Poincar\'e disc model.
	All line segments are geodesics with respect to the hyperbolic geometry.}
	\label{fig:regular_tessellation}
\end{figure}

\subsubsection{Symmetries of regular tessellations}
Regular tessellations are highly symmetric.
The group of symmetries of a regular tessellation is generated by reflections along lines of symmetry.
These lines of symmetry subdivide the tessellation into identical triangles with internal angles~$\pi/2$, $\pi/r$ and~$\pi/s$.
The symmetry group acts freely and transitively on the triangles, meaning that no triangle is stabilized by the group action and every triangle can be mapped onto any other.
By fixing one arbitrary triangle and assigning it the identity element of the group, we have a one-to-one correspondence between the triangles and the group elements.
Each triangle is uniquely determined by a triple consisting of a vertex, an edge and a face of the tessellation (all adjacent to one another).
Such triples are also known as \emph{flags}.
Two triangles are adjacent if and only if their corresponding flags differ in exactly one of their entries.

All symmetries of the tessellation must preserve the distance which makes them isometries.
Let $G_{r,s}$ denote the symmetry group of a hyperbolic tessellation with Schl\"afli symbol~$\{r,s\}$, then we have that $G_{r,s}\leq \operatorname{Isom}(\mathbb{H}^2)$.

The symmetry group of a regular tessellation~$\{r,s\}$ can be defined as a finitely presented group in terms of the generators and their relations.
There are three generators: $a_0$, $a_1$ and~$a_2$, each a reflection on one of the sides of a triangle.
The generator~$a_0$ reflects on the side of the triangle opposed to the vertex of the tessellation,~$a_1$ reflects on the side of the triangle opposed to the edge of the tessellation and~$a_2$ reflects on the side of the triangle opposed to the middle of the face of the tessellation.

We can define~$G_{r,s}$ as a finitely presented group in terms of its generators and the relations they obey.
As the generators are reflections we must have that $a_i^2$ is equal to the neutral element of the group.
Furthermore, the product of two generators forms a rotation.
For example, $a_0 a_1$ is a rotation around the center of a face which has order~$r$.
Similarly,~$a_1 a_2$ and~$a_0a_2$ are rotations around a vertex and an edge so that their orders are~$s$ and~$2$, respectively.
These are in fact all relations so that the symmetry group of a $\{r,s\}$ tessellation is
\begin{align*}
G_{r,s} = \left\langle a_0, a_1, a_2 \mid a_0^2, a_1^2, a_2^2, (a_0\, a_1)^{r}, (a_1 a_2)^s, (a_0 a_2)^2 \right\rangle .
\end{align*}
There exists a index 2 subgroup in $G_{r,s}$ which is generated by the rotations $\rho = a_0 a_1$ and $\sigma = a_1 a_2$.
The group is the subgroup of orientation preserving symmetries and denoted
\begin{align*}
	G_{r,s}^+ = \left\langle \rho, \sigma \mid \rho^{r}, \sigma^s, (\rho \sigma)^2 \right\rangle .
\end{align*}
The groups~$G_{r,s}^+$ are a special case of Fuchsian triangle groups, which are also denoted~$\Delta(2,r,s)$ in the literature.
Note that from the discussion in \Cref{sec:hyperbolic_isometries} it follows that
\begin{align}\label{eqn:triangle_group_in_PSL}
	G_{r,s}^+ \leq \operatorname{Isom}^+(\mathbb{H}^2)\simeq \operatorname{PSL}(2,\mathbb{R}) .
\end{align}

\subsubsection{Compact hyperbolic surfaces}\label{sec:hyperbolic_compact}
We  can  use  the  identification  between  the  fundamental triangles  and  the  group  elements  to  obtain  tessellations  of closed  surfaces.  
The  idea  is  to  consider  finite  quotients  of the infinite group, which leave the local structure of the group invariant. 
Geometrically, the procedure essentially consists of finding  translations  and  identifying  points  which  differ  by these translations. 
For example, on the 2D euclidean plane we can take an arbitrary translation and by identifying all points differing  by  this  translation  we  obtain  a  cylinder  of  infinite length. 
Taking a second translation, which is not co-linear with the first one, we obtain a torus.

In euclidean space all translation commute, but this is not true any longer in curved spaces, in particular in~$\mathbb{H}^2$.
Therefore we have to find a set of translations which, as a whole, commutes with all other elements of the symmetry group~$G_{r,s}$.
In the language of group theory; we are looking for a normal subgroup~$N$ of~$G_{r,s}$ which operates without fixed-points on~$\mathbb{H}^2$.
In the simplest case we consider a single translation~$\tau$ and take the normal closure of the group generated by~$\tau$ with respect to either~$G_{r,s}$ or~$G_{r,s}^{+}$ depending on whether we want to restrict ourselves to orientation preserving isometries:
\begin{align}\label{eqn:compactification_translation}
	N_\tau^{(+)} = \langle \tau \rangle^{G_{r,s}^{(+)}} = \lbrace g^{-1} t g \mid g\in G_{r,s}^{(+)},\, t \in \langle \tau \rangle  \rbrace 
\end{align}
If~$N_\tau$ has finite index in~$G_{r,s}^{(+)}$ then~$\mathbb{H}^2/N_\tau^{(+)}$ is a closed surface with finite area.

Such surfaces give rise to quantum codes with a linear encoding rate~\cite{breuckmann2016constructions}.


\section*{Acknowledgment}
NPB acknowledges support through his UCLQ Fellowship and the EPSRC Prosperity Partnership in Quantum Software for Simulation and Modelling (EP/S005021/1).
NPB would like to thank Johannes Keller, Jascha Ulrich and Tobias K\"uhn for many illuminating discussions on fiber bundles. JNE warmly thanks Wolfgang Soergel and Britta Kaisers.

\ifCLASSOPTIONcaptionsoff
  \newpage
\fi



\bibliographystyle{IEEEtran}
\bibliography{IEEEabrv,bibliography.bib}

%
%

%

%
%
%

\vfill

\begin{IEEEbiographynophoto}{Nikolas P. Breuckmann}
	holds a UCLQ Research Fellowship at University College London. He is interested in quantum information and related fields. He obtained his PhD at RWTH Aachen University working with Prof. Barbara Terhal on quantum fault-tolerance and quantum complexity theory. He has worked in industry at PsiQuantum, a Bay Area based start-up building a silicon-photonics based quantum computer.
\end{IEEEbiographynophoto}

\begin{IEEEbiographynophoto}{Jens N. Eberhardt} is a postdoctoral fellow at the University of Bonn. He is interested in geometric representation theory and its applications. He obtained his PhD at the University of Freiburg working with Prof. Wolfgang Soergel.  His previous stations include RWTH Aachen University, UCLA and the Max Planck Institute for Mathematics in Bonn. 
\end{IEEEbiographynophoto}




\end{document}